\documentclass[final,3p,times,12pt]{elsarticle}

\usepackage{epsfig}
\usepackage{amssymb}
\usepackage{amsthm}
\usepackage{amsmath}
\usepackage{multirow}
\usepackage{bm}
\usepackage{color}
\usepackage{booktabs}
\usepackage{subcaption}
\usepackage{graphicx}

\usepackage{float}
\usepackage{mathrsfs}
\usepackage{setspace}
\usepackage{epstopdf}
\usepackage{threeparttable}
\usepackage{pifont}

\newtheorem{thm}{Theorem}

\newtheorem{lem}{Lemma}

\newtheorem{definition}{Definition}

\begin{document}
	
\begin{frontmatter}
		
\title{Fault Diagnosis of Nonlinear Systems Using a Hybrid-Degree Dual Cubature-based Estimation Scheme}
		
\author[CU]{Yanyan Shen}
\ead{yy.shen1989@gmail.com}
		
\author[CU]{Khashayar Khorasani\corref{cor1}}
\ead{kash@ece.concordia.ca}
		
\cortext[cor1]{Corresponding author: Electrical and Computer Engineering, Concordia University, Montreal,	Quebec, H3G 1M8 Canada. K. Khorasani would  like to acknowledge the support received from
the Natural Sciences and Engineering Research Council of Canada (NSERC)
and the Department of National Defence (DND) under the Discovery Grant
and DND Supplemental Programs.}
		
\address[CU]{Electrical and Computer Engineering, Concordia University, Montreal, Quebec, H3G 1M8 Canada}
		
\begin{abstract}
			
In this paper, a novel hybrid-degree dual estimation approach based on cubature rules and cubature-based nonlinear filters is proposed for fault diagnosis of nonlinear systems through simultaneous state and time-varying parameter estimation. Our proposed dual nonlinear filtering scheme is developed based on \textit{case-dependent} cubature rules that are motivated by the following observations and facts, namely \textit{(i)} dynamic characteristics of nonlinear system states and parameters generally are distinct and posses different degrees of complexities, and \textit{(ii)} performance of cubature rules depend on the system dynamics and vary due to handling  of high-dimensional integrations approximations. For improving the robustness capability of our proposed methodologies a modified cubature point propagation method is incorporated. The performance of our proposed  dual estimation strategy is demonstrated and evaluated by application to a nonlinear gas turbine engine for addressing the component fault diagnosis problem within an integrated fault detection, isolation and identification framework. Robustness analysis is implemented to verify the capability of our proposed approaches to deal with parametric uncertainties and unmodeled dynamics. Extensive simulation case studies and discussions with respect to component fouling, erosion or abrupt faults are provided to substantiate and justify the superiority of our proposed fault diagnosis methodology when compared to other well-known alternative diagnostic techniques such as the Unscented Kalman Filters (UKF) and Particle Filters (PF) that are commonly available in the literature.
			
\end{abstract}
				
		
\end{frontmatter}
	
Estimation as a quantitative evaluation process of unmeasured states and/or parameters from uncertain or imprecise observations is a fundamental problem for nonlinear systems in various disciplines such as control and fault diagnosis \cite{afshari2017gaussian}. Model-based fault diagnosis (FD) techniques (that consist of fault detection, isolation, and identification (FDII)) relying on estimation approaches have been extensively investigated, albeit mostly  for linear systems using e.g.,  Kalman filters (KF) \cite{zhang2018adaptive} and observer-based methods \cite{gao2008novel}. Although linear approaches enable one to achieve acceptable estimation performance locally, their performance undergo deterioration as nonlinear dynamics dominate the system behavior and risk losing convergence given an accurate approximation requirement. As far as fault diagnosis of nonlinear systems are concerned linear approaches might be subject to high rates of false alarms and poor detection and diagnosis performance.

A great deal of investigation on fault detection and isolation (FDI) problems that utilize nonlinear estimation approaches have been conducted in the literature.  According to the employed nonlinear filters, these approaches can be broadly categorized into: 1) FDI methods developed using local nonlinear filters, e.g., Extended Kalman Filters (EKF) in \cite{foo2013sensor}, Unscented Kalman Filters (UKF) in \cite{shang2011sensor}, Cubature Kalman Filters (CKF) in \cite{kim2016cubature}, Gauss-Hermite Filters (GHF) in \cite{wang2016novel}; and 2) FDI methods developed using  global nonlinear filters, e.g., Particle Filters (PF) in \cite{yin2015intelligent} and Ensemble Kalman Filters (EnKF) in \cite{daroogheh2018ensemble}. These work have introduced basic frameworks for state estimation of nonlinear systems, and FDI schemes using the corresponding generated residuals.

In this paper, component faults are represented as variations of health-related parameters, and their diagnosis will be achieved through simultaneous state and parameter estimation. The central idea of using simultaneous state and parameter estimation to accomplish fault diagnosis of nonlinear system is certainly not new. One significant set of publications focus on joint estimation schemes. The parameters and states are augmented into one vector for performing  simultaneous estimation using e.g. nonlinear observer-based methods \cite{buciakowski2017bounded}, EKF-based methods \cite{stojanovic2016joint}, and UKF-based methods \cite{yu2017lithium}. Theoretically, joint estimation scheme casts the simultaneous estimation problem into a single filtering scheme; however, the consequent drawback is the large and high-dimensional matrix operation of the augmented model.

Another set of publications focus on simultaneously estimating states and parameters using parallel filters, namely through dual estimation schemes. A higher accuracy can be expected as the estimation is performed in a closed-loop manner \cite{wei2017multi}. Additionally, Shi \textit{et al.} in \cite{shi2019sequential} have shown that  the computational time of the dual EKF-based scheme is 25\% less than the joint one. However, due to the linearization errors, the dual EKF-based scheme might fail to converge to accurate estimation in cases of highly nonlinear systems. Plett in \cite{plett2006sigma} developed a dual UKF-based scheme for battery management systems with better  performance as compared to EKF.

Nevertheless, when the system order exceeds beyond three, perturbations are induced which cause the numerical inaccuracy; and moreover presence of negative weights might risk the filtering stability \cite{arasaratnam2009cubature}. In \cite{daroogheh2017dual}, a dual PF-based fault diagnosis scheme was developed for a single-spool gas turbine engine. The superiorities on higher estimation accuracy and lower false alarms have been verified with guaranteed filtering stability. Other PF-based dual estimation works are found in \cite{ding2017fault,olaizola2019real} for different applications. Although their estimation performance have been justified to be highly accurate, their significant computational demands challenge utilization in wide range of real-time applications. One can observe from the above dual estimation works that the employed nonlinear filters are critical to the performance of simultaneous estimation and fault diagnosis problems.

The CKF has recently been extensively studied for high-dimensional state estimation. Its derivative-free properties, reasonable estimation accuracy when subjected to Gaussian considerations, divergence avoidance, and dimensionality issues are the main advantages when compared to EKF, UKF, or PF estimation techniques \cite{afshari2017gaussian,arasaratnam2009cubature}. It has also been successfully applied in various fields such as lithium-ion batteries state of charge estimation \cite{peng2019improved}, missile attitude estimation \cite{liu2019maximum}, among others. Few studies can  be found that take advantage of the CKF state estimation capability for FDI problems. For example, Kim \textit{et al.} in \cite{kim2016cubature} used a CKF to generate residuals to perform FDI of a multi-unmanned vehicle; Xu \textit{et al.} in \cite{xu2017comparative} has run various experiments to verify that CKF has the best diagnosis performance as compared to the EKF and strong tracking filter, and is more suitable for FDI of ECAS systems.

However, CKF has limited estimation accuracy for certain nonlinearities due to its employed  \textit{3rd-degree} cubature rules. Consequently, biased residual signals might be generated  which can lead to misclassifying the health status. To remedy this drawback, and as a motivation to enhance the level of accuracy and to capture higher nonlinear dynamics, higher-degree cubature rules are proposed to be utilized. Jia \textit{et al.} in  \cite{jia2013high} designed a high-degree CKF based on the spherical-radial rule. Its estimation accuracy has been shown to be superior to UKF, CKF, and PF. The downside with respect to the higher degree of cubature rules are the increase in both the computational cost and the risks in numerical instability. Consequently, investigating efficient and numerically stable cubature rules, and incorporating such cubature rules within the Bayesian filtering framework to propose our   ``\textit{cubature-based nonlinear filters} (CNF)'' and to solve the nonlinear FD estimation problem is one of the main objectives of this paper.

Derivation of cubature rules have  ascribed great importance in the field of numerical mathematics during the past few decades, e.g., \cite{wang2013spherical,ballreich2019stable}. Although a variety of cubature rules have been developed in numerical mathematics, not all of them are directly applicable to the general nonlinear filtering problems given that many were developed to only solve certain specific problems. Our focus in this paper is on those cubature rules that can either be generalized to fault diagnosis and health monitoring problems of nonlinear systems having arbitrary order as well as arbitrary degrees of polynomial accuracy and that can particularly perform well for the targeted nonlinear gas turbine engine (GTE) application.

Early research on FDII of GTE have been one of the challenging application areas that have received much attention. Gas path analysis (GPA) is one of the most popular diagnostic procedures which relies on discernible changes in the observable parameters of the engine to detect presence of faults. Various fault diagnosis techniques have been developed for GTE based on GPA ranging from Kalman filters-based variants \cite{pourbabaee2015sensor,daroogheh2018ensemble}, neural networks \cite{amozegar2016ensemble}, data-driven methods \cite{naderi2018data}, and component adaptation method \cite{tsoutsanis2014component}.

To the authors' best knowledge, the fault identification problem along with a unified integrated fault diagnosis scheme that employs CNF based on a dual estimation scheme has not been investigated in the literature. This represents as another  objective of this paper where our goal  is to develop an efficient and stable dual cubature-based estimation scheme to not only detect and isolate component faults but also to accurately identify  severity of simultaneous multi-mode scenarios for the GTE system.

However, several limitations and open areas are still outstanding for applying CNF to address component FD problems. Specifically, one can state the following challenges in and shortcomings of the literature:
\begin{itemize}
	\item[1)] The balancing selection and choice is between reasonable estimation accuracy and acceptable computational cost for a real-time implementation application. The cubature rules utilized in our corresponding CNF determine the estimation performance through the capacity of capturing the system nonlinear dynamics. Improved accuracy can be achieved for strong nonlinearities by using higher rule degrees, however, the computational demands would be growing quadratically or even of higher order.
	\item[2)] Numerical stability of certain CNF decreases with the increasing degree of accuracy and the increasing system order due to the growing influence of negative weights. The accumulative impacts of numerical instabilities through iterative calculations do indeed impose high risks that lead to unstable filtering solutions \cite{ballreich2019stable}.
\end{itemize}

Moreover, in dual estimation research that investigated Dual-EKF, Dual-UKF or Dual-CKF one actually utilizes the \textit{same order} of accuracy for the state and parameter estimation modules. However, due to the fact that the dynamics of states and parameters are completely different, higher false alarms or low accuracy rates can occur given their incapability of capturing the corresponding different degrees of nonlinearities.

Motivated by the above discussion on FD of nonlinear systems, in this paper a novel hybrid-degree dual estimation scheme is  implemented for the first time in the literature for real-time health monitoring of GTE. The term ''\textit{hybrid}'' indicates and refers to the notion that different degrees or cubature rules are considered for developing the nonlinear state and parameter estimation schemes. The hybrid-degree dual estimation scheme is motivated by the fact that the nonlinear dynamics of the system states and parameters in general, and the GTE in particular, are practically completely different. For example, the GTE  state dynamics contain higher degrees of nonlinearities where a conventional \textit{3rd-degree} cubature rule cannot completely capture and represent, thereby necessitating one to utilize a higher-degree cubature rule, whereas the GTE health-related parameter dynamics are in effect less complex so that a high-degree cubature rule is not necessary given the practical  implementation and limited computational resources available in real-time.

To summarize the main contributions of this paper can be stated as follows:
\begin{itemize}
	\item[1)] A novel stable and efficient hybrid-degree dual cubature-based filtering approach for FD of nonlinear systems. In contrast to the same degree-based Dual-EKF, Dual-UKF, and Dual-CKF, our proposed case-dependent hybrid degree solution will improve the estimation accuracy and FD performance. Furthermore, our proposed hybrid-degree strategy has the flexibility to simultaneously achieve improved accuracy and computational efficiency by considering the prior knowledge of  system dynamics and  user's specifications and requirements. We have also performed quantitative comparative evaluation and analysis of various approaches  available in the literature that are utilized as reference benchmark.
	\item[2)] Compared to Dual-UKF, Dual-CKF, and Dual-PF algorithms, advantages and superiorities of our proposed FD scheme are justified and validated in terms of fault detection promptness, isolation, and identification accuracy, false alarm rates, precision, and computational cost.
	\item[3)]	Robustness capabilities of the proposed methodology with respect to modeling uncertainties have been formally and quantitatively analyzed. For improving the robustness and reliability of the proposed methodology against parametric, unmodelled dynamic uncertainties, a modified cubature points propagation methodology is incorporated into the proposed framework. Comparative  evaluations in terms of false alarm rates, fault detection time, and accuracy performance metrics are provided.
	\item[4)] Performance of the proposed unified component FD framework is verified and validated by application to multi-mode simultaneous/concurrent fault diagnosis problem of GTE system that is subject to both abrupt and incipient fault types.
	\item[5)] The boundedness properties of the estimated health-related parameter errors are formally investigated and analyzed.
\end{itemize}

The remainder of this paper is organized as follows. In Section \ref{sec2}, the statement of dual estimation problem for system states and parameters based on numerical cubature rules is presented. The cubature-based nonlinear filters (CNF) are developed in Section \ref{sec3}, and a detailed design procedure of our proposed hybrid-degree dual estimation scheme is provided in Section \ref{sec4}. The state and parameter estimation problems as well as FD strategy formulations are  provided. The effectiveness of the proposed framework is verified by its application to a component FD problem of a GTE system. Comparative studies are conducted in Section \ref{sec5} where performance of our proposed CNF and FD strategies are evaluated in terms of metrics of accuracy, stability factor, and computational cost. Conclusions  are provided in Section \ref{sec6}.

\section{Problem Statement}
\label{sec2}

Consider the general discrete-time nonlinear system
\begin{equation}
\label{eq:nonlinear_sys_state}
x_{k+1}=f(x_{k},\theta_{k},u_k)+w_{k}
\end{equation}
\begin{equation}
\label{eq:nonlinear_sys_measure}
z_k=g(x_k, \theta_{k},u_k)+v_{k}
\end{equation}
where $x_k \in \mathbb{R}^{n_x}$, $z_k \in \mathbb{R}^{n_z}$, $\theta_k \in \mathbb{R}^{n_\vartheta}$ denote the system states, measurements, and health-related parameters, respectively. Also, $u_k \in \mathbb{R}^{n_u}$ denotes the control input, $f: \mathbb{R}^{n_x} \times \mathbb{R}^{n_\theta} \times \mathbb{R}^{n_u} \rightarrow \mathbb{R}^{n_x}$ denotes the nominal nonlinear  system dynamics, $g: \mathbb{R}^{n_x} \times \mathbb{R}^{n_\theta} \times \mathbb{R}^{n_u} \rightarrow \mathbb{R}^{n_z}$ denotes a known nonlinear function,
$w_{k}$ and $v_{k}$ represent the zero-mean uncorrelated Gaussian white noise sequences for states and measurements, with  $\mathbb{E}[w_{k}w_{l}^T]=\Sigma_{w,k}\delta_{k,l}$, $\mathbb{E}[v_{k}v_{l}^T]=\Sigma_{v,k}\delta_{k,l}$ and $\mathbb{E}[w_{k}v_{l}^T]=0$, respectively.

It is assumed that the dynamic characteristics of the multiplicative health-related parameters are represented by
\begin{equation}
\label{eq:nonlinear_sys_par}
\theta_k=h(\theta_{k-1})+ \tau_k
\end{equation}
where $h$ denotes the degradation dynamics, and $\tau_{k}$ represents a zero-mean Gaussian white noise with $\mathbb{E}[\tau_{k}\tau_{l}^T]=\Sigma_{\tau,k}\delta_{k,l}$.
\vspace{1mm}\\
\textbf{Problem Statement 1:} This paper aims to develop a unified multi-mode FD methodology that simultaneously handles fault detection, isolation and identification problems. Our goals for the proposed FD methodology are to provide a fast detection, low false alarm rates and missed detections, and low estimation errors, while being computationally feasible for real-time implementation. This is accomplished by monitoring the status of health-related parameters $\theta_k$ by developing an efficient dual estimation scheme. 

The objective of dual estimation scheme here can be formulated as that of approximating the conditional expectations with respect to states and parameters as governed by
\begin{equation}
\label{eq:obj_state}
\mathbb{E}(\phi_1(x_k)|z_{1:k},\theta_{k-1})=\int_{\mathbb{R}^{n_x}} \phi_1(x_k)p(x_k|z_{1:k},\theta_{k-1})dx_k
\end{equation}
\begin{equation}
\label{eq:obj_par}
\mathbb{E}(\phi_2(\theta_k)|z_{1:k},x_k)=\int_{\mathbb{R}^{n_\theta}} \phi_2(\theta_k) p(\theta_k|z_{1:k},x_k)d\theta_k
\end{equation}
where $\phi_1(x_k)$ and $\phi_2(\theta_k)$ are functions that are to be simultaneously estimated, and $z_{1:k}=\{z_1,z_2,\cdots,z_k\}$ denotes the available observations up to the time instant $k$, $p(x_k|z_{1:k},\theta_{k-1})$ and  $p(\theta_k|z_{1:k},x_k)$ denote the conditional probability density functions (pdfs) that are expected to be approximated by the developed nonlinear filters. 

The key difficulty for obtaining the pdfs by using Bayesian filtering scheme is the involved intractable multivariate integrals in both prediction and update stages. Proceeding with the assumption that the pdfs of  states and parameters are Gaussian, only the moments of their means and covariances are needed to be computed. The computation of the mean and covariance requires multivariate integrals that are of the tractable Gaussian weighted form. 

The underlying problem now is to accurately approximate the Gaussian weighted integrals in each stage of the Bayesian filtering. 
Let us take $\phi_1(x_k)=x_k$ and $\phi_2(\theta_k)=\theta_k$, and denote the prior Gaussian distributions of states and parameters at time $k-1$ as $\mathcal{N}_{k-1}^x\triangleq(x_{k-1};\hat{{x}}_{k-1}, P_{k-1}^{xx})$ and $\mathcal{N}_k^\theta\triangleq(\theta_k;\hat{{\theta}}_{k-1}, P_{k-1}^{\theta\theta})$, respectively, where $\hat{{x}}_{k-1}$ and $\hat{{\theta}}_{k-1}$ denote the estimated mean, $P_{k-1}^{xx}$ and $P_{k-1}^{\theta\theta}$ represent the corresponding covariance matrices. Let us consider the Gaussian integrals for the mean of predictive density of states and parameters in the prediction stage where they can be approximated by cubature rules as:
\begin{equation}
\label{eq:state_filter_1}
\begin{split}
\int_{{\mathbb{R}}^{n_x}}\phi_1(x_k)\mathcal{N}_{k-1}^xdx_{k-1}
& \approx \sum_{i=1}^{N_{d_x}} {\phi_1}(\tilde{\xi}_{i, k|k-1}^{d_x}) {w}_i^{d_x}
\end{split}
\end{equation}
\begin{equation}
\label{eq:par_filter_1}
\begin{split}
\int_{{\mathbb{R}}^{n_\theta}}\phi_2(\theta_k)\mathcal{N}_{k-1}^\theta d\theta_{k-1}
&\approx \sum_{j=1}^{N_{d_\theta}} {\phi_2}(\tilde{\xi}_{j, k|k-1}^{d_\theta}) {w}_j^{d_\theta}
\end{split}
\end{equation}
where the variables $w$, $N$ and $\tilde{\xi}_{k|k-1}$ denote the weights, the total number of points, and the propagated cubature points at time $k|k-1$ based on the sampled points ${\xi}$. The subscripts or superscripts $d_x$ and $d_\theta$ represent the degree of cubature rules for state and parameter filters, respectively. 

Let us define two sets $\{\xi_{i}^{d_x},{w}_i^{d_x}\}$, $i=1,\cdots, N_{d_x}$ and $\{\xi_{j}^{d_\theta},{w}_i^{d_\theta}\}$, $j=1,\cdots, N_{d_\theta}$. Obviously, design of these variables affect the approximation performance of (\ref{eq:state_filter_1}) and (\ref{eq:par_filter_1}). They are determined according to the degree of cubature rules, $d_x$ and $d_\theta$ according to the employed cubature rules. 

This paper designs the degree of cubature rules according to the prior knowledge on  $\phi_1(\cdot)$ and $\phi_2(\cdot)$ functions.
Specifically, the function in (\ref{eq:state_filter_1}) can be obtained by the \textit{$d^x$th-degree} cubature rules if it is exact for the nonzero function $\phi_1(x_k)$ whose components are linear combinations of monomials having coefficients $a_{\bm{\alpha}}$ and monomials $\bm{x}^{\bm{\alpha}}= \prod_{i=1}^{n_x}x_i^{\alpha_i}$, with the total degree up to $d^x$. In other words, the monomials integers   in  state dynamics satisfy the following
\begin{equation}
\label{eq:dx}
d^x= \text{max} \{\sum_{i=1}^{n_x}|\alpha_i|: a_{\alpha} \neq 0\}
\end{equation}

In case of $\alpha_1+\cdots+\alpha_{n_x} >d^x$ in $\phi_1(\cdot)$, the \textit{$d^x$th-degree} cubature rules cannot be exactly approximated by the monomials with the \textit{$d^x$th-degree} accuracy, since the \textit{$d^x$th-degree} cubature rule can reach a $d^x$th-degree of accuracy \cite{jia2013high}. Similarly, the parameter estimation is accomplished by the $d^{\theta}$th-degree cubature rules if the parameter dynamics fulfill $d^\theta= \text{max} \{\sum_{j=1}^{n_\theta}|\beta_j|: b_{\beta} \neq 0\}$, where $\bm{\theta}^{\bm{\beta}}= \prod_{i=1}^{n_{\theta}}\theta_j^{\beta_j}$ denotes the monomials. From approximation accuracy point of view, the choice of degrees are essentially case-dependent given the priori knowledge on the underlying system dynamics.
\vspace{1mm}
\\
\textbf{Problem Reformulation:} Construct appropriate sets  $\{\xi_{i}^{d_x},{w}_i^{d_x}\}$ and $\{\xi_{j}^{d_\theta},{w}_j^{d_\theta}\}$, $i=1,\cdots, N_{d_x}$, $j=1,\cdots, N_{d_\theta}$ to approximate the involved Gaussian weighted integrals (\ref{eq:state_filter_1}) and (\ref{eq:par_filter_1}) in the dual estimation scheme  by proposing proper $d_x$th-degree and  $d_\theta$th-degree cubature rules for state and  parameter estimation problems, respectively. 

Importantly, the designed case-dependent degree of cubature rules and the employed cubature formulas will eventually lead to different CNFs. For achieving reliable fault diagnosis performance though designing efficient, accurate and stable dual-CNF scheme, this paper will not limit to one solution, but provide comparisons and evluations on various choices.  In this case, we can provide a generalized solution to arbitrary nonlinear system order with arbitrary degree of accuracy, or that can particularly perform quite satisfactorily on specific cases of nonlinear systems such as the GTE for addressing the fault diagnosis problem. 

With respect to the application to GTE system, the unmodelled dynamics, model mismatches, parametric uncertainties, and noise discrepancy between the actual GTE and the on-board engine model (OBEM) can increase the fault detection time and lead to occurrence of false alarms and incorrect fault severity estimation results. The details are provided in Section IV.D. In order to verify the robustness of our proposed methodology in presence of measurement uncertainties, robustness analysis of the proposed CNF and FD methodologies on fault estimation performance against them is conducted, based on the Assumption 1 below.
\vspace{1mm}\\
\textbf{Assumption 1:} The measurement uncertainty is bounded by  $\zeta(x_k,u_k) \le \bar{\zeta}$, for $\bar{\zeta}>0$, which is present in the measurement equation as below:
\begin{equation}
\label{eq:nonlinear_sys_measure_uncertain}
z_k=g(x_k, \theta_{k},u_k)+\zeta(x_k,u_k)+v_{k}
\end{equation}
Acceptable ranges on modeling uncertainties that do not lead to false alarms will be specified   under various healthy and faulty scenarios in Section V.

\section{Cubature-Based Nonlinear Filters (CNFs)}\label{sec3}

The goal of this section is twofold. First, a class of cubature rules are presented in an accessible manner to aid  in  implementation of these methods. Second, a class of CNFs is constructed based on cubature rules for the nonlinear estimation problem.

\subsection{The Construction of Sets $\{\xi_{i}^{d_x},{w}_i^{d_x}\}$ and $\{\xi_{j}^{d_\theta},{w}_j^{d_\theta}\}$}
\label{sec3-1}

As stated in Section \ref{sec2}, the objective of the dual estimation problem can be transformed into designing the sets $\{\xi_{i}^{d_x},{w}_i^{d_x}\}$ and $\{\xi_{j}^{d_\theta},{w}_j^{d_\theta}\}$. This section presents a class of cubature rules where for sake of generality an $n$-dimensional nonlinear estimation problem is considered and the general set $\{\xi_{i}^{d},{w}_i^{d}\}$ will be constructed. The integral with respect to a general Gaussian distribution   $\mathcal{N}(x;\hat{x},P)$ is approximated by
\begin{equation}
\label{eq:integral}
\mathcal{I}(\phi)=\int_{\mathbb{R}^{n}} \phi(x)\mathcal{N}(x;\hat{x},P)dx \approx \sum_{i=1}^{N} w_i \phi(S\xi_{i}+\hat{x})
\end{equation}
where $\phi$ denotes an arbitrary nonlinear function, and $P=SS^T$. The weight $w_i$, the sampled cubature points $\xi_i$ and the total number of points $N$ are determined  and specified  in the following subsections by analyzing various cubature rules.

Generally, the choice of cubature points $\xi_{i}^{d}$ depends on the domain of integration. This paper investigates the cubature rules over both spherical surfaces and the entire $n$-D space surfaces. Of interest is the spherical surface that is concerned with both Genz's theorem \cite{genz2003fully} and the Mysovskikh's theorem \cite{mysovskikh1980approximation} for obtaining the corresponding cubature rules.

\subsubsection{\underline{Cubature Rules Over Spherical Surface}}

The integral in (\ref{eq:integral}) can be transformed into the following spherical-radial coordinate system
\begin{equation}
\label{eq:SR_coordinate}
\mathcal{I}(\phi)= \int_{0}^{\infty} \int_{U_n} \phi(r \textbf{s}) r^{n-1} \text{exp} (-r^2) d \sigma dr
\end{equation}
where the spherical surface $U_n=\{\textbf{s} \in \mathbb{R}^n : \textbf{s}^T \textbf{s}=1 \}$, with $x=r \textbf{s}$
and $r=\sqrt{x^T x}$, $\textbf{s}=[\textbf{s}_1, \cdots, \textbf{s}_n]^T$, and $\sigma(\cdot)$ denotes the spherical surface measure. The integrals in  (\ref{eq:SR_coordinate}) can be addressed by separately approximating two sub-integrals, namely \textit{(a)} the spherical integral $\int_{U_n} \phi_s(\textbf{s}) d \sigma(\textbf{s})$, and \textit{(b)} the radial integral $\int_{0}^{\infty} \phi_r(r) r^{n-1} \text{exp}(-r^2) dr$ \cite{arasaratnam2009cubature}.

$\bullet$ {\textit{\underline{The Spherical Rule}}}: It is utilized to solve the spherical integral as $\sum_{p=1}^{N_s} w_{s,p} \phi(\textbf{s}_p)$, where $\textbf{s}_p$ and $w_{s,p}$ denote spherical points and weights.
We concentrate on two rules to compute the spherical integral, namely: \textit{(i)} \textit{Genz spherical rule} \cite{genz2003fully} that allows a system with an arbitrary order achieves an arbitrary degree of accuracy, and \textit{(ii)} \textit{Mysovskikh spherical rule} \cite{mysovskikh1981interpolatory} that enables one to realize a more efficient approximation when compared to the Genz spherical rule for systems having order $n \ge 4$. Specifically, the Genz method constructs \textit{$(2m+1)$th-degree} spherical rule over the surface of the sphere $U_n$, where $m= \rho_1 + \rho_2 +\cdots + \rho_n$, with $\rho_p$ denoting nonnegative integers.
For designing the $d$th-degree spherical rule, one sets $d=2m+1$ and analyze each non-negative integer. Mysovskikh in \cite{mysovskikh1980approximation} derived a rule based on the transformation group of regular simplex with vertices $\textbf{a}^{(p)}=[a_{1}^{(p)},a_{2}^{(p)},\cdots,a_{n}^{(p)}]^T$, $p=1,2,\cdots,n+1$. For designing the $d$th-degree spherical rule, various topologies are considered, as provided in Algorithm 1.

$\bullet$ {\textit{\underline{The Radial Rule}}}: It aims to solve the radial integral in  (\ref{eq:SR_coordinate}) as $\sum_{q=1}^{N_r} w_{r,q} \phi({r}_q)$, where ${r}_q$ and $w_{r,q}$ denote radial points and weights. Moment matching method is employed for computing the radial rule. The key idea is to determine radial points and weights that satisfy the moment equations based on the rule degree and the system order. For more details refer to \cite{arasaratnam2009cubature}.

Consequently, (\ref{eq:SR_coordinate})  can further be formulated as
\begin{equation}
\label{eq:sphe_radial_cubature}
\begin{split}
\mathcal{I}(\phi)&
\approx \sum_{q=1}^{N_r} \sum_{p=1}^{N_s} w_{r,q} w_{s,p} \phi(r_q \textbf{s}_p)
\end{split}
\end{equation}
where the sets $(\textbf{s}_p, w_{s,p})$ and $(r_q, w_{r,q})$ are obtained by using the spherical and radial rules, with $N_s$ and $N_r$ denotes the corresponding required number of points. respectively, as provided in Algorithm 1. 

\begin{definition}
	\label{def1}
	The set $\{\xi_i^{d},{w}_i^{d}\}$ based on the spherical surface by using the spherical-radial cubature rules is constructed as
	\begin{equation}
	\label{eq:set1}
	\begin{split}
	\xi_i^{d}&\in \{\sqrt{2r_q \textbf{s}_p}, p\!=1,\!\cdots\!,N_s, q=1,\!\cdots\!, N_r\}\\
	{w}_i^{d}&\in \{w_{r,q}  w_{s,p} /\pi^{n/2}, p\!=1,\!\cdots\!,N_s, q=1,\!\cdots\!, N_r\}
	\end{split}
	\end{equation}
	where $i=1,\cdots N$, with $N=N_r \times N_s$ if $r_q \neq 0$, and $N=(N_r-1) \times N_s+1$ if one of $r_q$ in (\ref{eq:sphe_radial_cubature}) is zero.
\end{definition}
\begin{thm} [\cite{moller1979lower}]
	The number of nodes $N$ of a cubature of degree $d=2s-1$ satisfies $N \geq N_{min}$ with
	\begin{equation}
	\footnotesize
	\nonumber
	N_{min}\!=\!\!\left\{\!\!\begin{array}{ll}
	\left(\begin{array}{l}
	n\!+\!s-1 \\
	n
	\end{array}
	\right)\!+\!\sum_{k=1}^{n-1}2^{k-n}	\left(\begin{array}{l}
	k\!+\!s-1 \\
	k
	\end{array}
	\right), & $s$ \ even\\
	\left(\begin{array}{l}
	n\!+\!s-1 \\
	n
	\end{array}
	\right)\!+\!\sum_{k=1}^{n\!-\!1}(1\!-\!2^{k\!-\!n})	\left(\begin{array}{l}
	k\!+\!s-2 \\
	k
	\end{array}
	\right), & $s$ \ {odd}
	\end{array}
	\right.
	\end{equation}
\end{thm}
It is easy to check that for $d=3$, $N_{min}=2n$, while $d=5$ yields $N_{min}=n^2+n+1$. Given that one of our main goals is to investigate the FD problem of GTE system,  efficiency concerns in determining whether our approach is applicable to a practical problem is of significant importance. In this case, a special but efficiency \textit{5th-degree} cubature rule that is integrated over the entire n-D space is considered as follows.

\subsubsection{\underline{Efficient Cubature Rule Over the Entire n-D Space}}
Given that one of our main goals is to investigate the FD problem of GTE system,  efficiency concerns in determining whether our approach is applicable to a practical problem is of significant importance. In this case, a special but efficiency \textit{5th-degree} cubature rule is considered. 

\begin{definition}
	\label{def2}
	The set $\{\xi_i^{d},{w}_i^{d}\}$ based on n-D surface by using the fifth-degree modified Stroud's Theorem is constructed as
	\begin{equation}
	\label{eq:set2}
	\begin{split}
	\xi_i^{d}\in\{\ \underbrace{\bm{\nu}_1,\ -\bm{\nu}_1}_2, \ \ \underbrace{\bm{\nu}_2,\ -\bm{\nu}_2}_{2C_n^1=n},\underbrace{\bm{\nu}_3,\ -\bm{\nu}_3}_{2C_n^2\!=n(n\!-1)\!/2}\!\!\!\}\\
	\end{split}
	\end{equation}
	where $\bm{\nu}_1$,  $\bm{\nu}_2$, and $\bm{\nu}_3$ denote the points that are related to the system order $n$ that have been given below in Algorithm 1, for $i=1,\cdots N$ with $N=n^2+n+2$. For further detail on the weights ${w}_i^{d}$ and all other coefficients refer to \cite{stroudapproximate}. 
\end{definition}

One downside of this rule is that it is only valid for a limited range of system orders of $2 \le n \le 7$, it is perhaps the most efficient one possible among the \textit{5th-degree} rules since it requires only \textit{one} point more than the lower bound that is given by \cite{moller1979lower}. A variety of dynamical systems fulfill such condition, including the state dynamics of our GTE system. Therefore, the efficient cubature rule that is derived over the entire $n$-D space as proposed in Stroud \cite{stroudapproximate} will also be discussed in this paper.

Finally, given the logic of determining the degrees \textit{$d^x$} and \textit{$d^\theta$} in Section \ref{sec2}, and the known dimensions of $n_x$ and $n_\theta$, the objective sets $\{\xi_i^{d_x},{w}_i^{d_x}\}$ and $\{\xi_j^{d_\theta},{w}_j^{d_\theta}\}$ could be consequently constructed by using the set $\{\xi_i^{d},{w}_i^{d}\}$.
\vspace{1mm}\\
{\small
	\textit{\textbf{Algorithm 1: Pseudo-code for designing $d$th-degree CR}}\\
	\textbf{Input}: Desired degree $d$, System dimension $n$: \\
	\textbf{Output}: $\mathcal{I}(\phi)$=\bf{CubatureRule}($d,n$): \\
	1: \ding{111} \textbf{Cubature rules over spherical surface using (\ref{eq:sphe_radial_cubature})}:  \\
	2: \quad \textbf{$\bullet$ $d$th-degree spherical rules $[s_p^{d},w_{s,p}^{d}]$}: \\				
	3: \quad \quad \ding{192} $[\bm{s}_p^{d},w_{s,p}^{d}]=\textbf{Genz}(n,d)$:  \\
	\textbf{4}: \quad \quad \quad \quad Determine all possible nonnegative $\bm{\rho}=[\rho_1,\cdots, \rho_n]$\\
	\textbf{5}: \quad \quad \quad \quad \quad with $|\rho|=m$ and $d=2m+1$\\
	\textbf{6}: \quad \quad \quad \quad For each possible $[\rho_1,\cdots, \rho_n]$, derive \cite{genz2003fully}\\
	\textbf{7}: \quad \quad \quad \quad \quad $\textbf{s}_p^d\!\!=\![v_1u_{\rho_1},\!\!\cdots\!,v_nu_{\rho_n}]^T$, \\
	\textbf{8}: \quad \quad \quad \quad \quad \quad with $u_{\rho_p}\!\!=\!\sqrt{\rho_p/m}$, $v_p\!\!=\! \pm 1$\\
	\textbf{9}: \quad \quad \quad \quad \quad Calculate the weight $w_{s,p}^d=2^{-c(\textbf{u}_p)}w_{\bm{\rho}}$\\
	\textbf{10}: \quad \quad \quad \quad  End\\
	11: \quad \quad  \ding{193} $[s_p^{d},w_{s,p}^{d}]=\textbf{Mysovskikh}(n,d)$: \\
	12: \quad \quad \quad \quad $\textbf{s}_p^d$: Transformed topologies of regular simplex with:	\\
	\textbf{13}: \quad \quad \quad \quad \ $[a_{p,1},\cdots,a_{p,n}],p\!=\!1,\cdots\!,n\!+\!1, i\!=\!1,\!\cdots\!,n$	\\
	\textbf{15}: \quad \quad \quad		\quad  	\
	{\footnotesize$
		a_{p,i}\!\triangleq \! \left\{\!\!\!\!\!\! \!\!\!\!\!
		\begin{array}{ccc}
		& \!-\!\sqrt{(n\!+\!1)/\left[n(n\!-\!i\!+\!2)(n\!-\!i+1)\right]},&i\!<\!p \\
		&\!	\!-\!\sqrt{(n\!+\!1)(n\!-\!p\!+\!1)/\left[n(n\!-\!i\!+\!2)\right]},&i\!=\!p\\
		&	0,&i\!>\!p
		\end{array}
		\right.
		$}\\
	\textbf{16}: \quad \quad \quad \ $w_{s,p}$: derived in \cite{mysovskikh1981interpolatory}\\
	17: \quad \textbf{$\bullet$ $d$th-degree radial rules}\\ 
	18: \quad \quad  $[s_q^{d},w_{r,q}^{d}]=\textbf{MomentMatching}(n,d)$ \cite{arasaratnam2009cubature} \\
	\textbf{19}: \quad \quad  Solve $\sum_{q\!=\!1}^{N_r} w_{r,q} \phi(r_q)\!=\!\frac{1}{2}\Gamma(\frac{n\!+l}{2})$, $l\!=\!0,2,\!\cdots\!,2m$ \cite{arasaratnam2009cubature}\\
	20:  \textbf{$\Rightarrow$ $\mathcal{I}(\phi)$=SphericalSurface $(n,d,s_q^{d},w_{s,q}^{d},s_p^{d},w_{r,p}^{d})$} using\\
	21: \quad \quad  \quad \quad\quad $\sum_{q=1}^{N_r} \sum_{p=1}^{N_s} w_{r,qq} w_{s,p} \phi(r_q \textbf{s}_p)$	\\
	22: \ding{111} \textbf{Efficient \textit{5th-degree} cubature rules over n-D space}\\
	23: \textbf{  $\mathcal{I}(\phi)$=EntireSurface $(n,\nu_1,\nu_2,\nu_3,w_1,w_2,w_3)$} \\
	24: \quad Re-defined: $\nu_1=[\eta,\eta,\cdots,\eta]$, $\nu_2=[\lambda,\xi,\cdots,\xi]$, \\
	25: \quad \quad \quad \quad \quad \quad
	$\nu_3=[\upsilon,\upsilon,\gamma,\cdots,\gamma]$,\\
	26: \quad \quad \quad   $\upsilon=(-3 \pm \sqrt{16-2n})\gamma$, $\gamma^2=\frac{3 \pm \sqrt{7-n}}{2(16-n\pm 4 \sqrt{16-2n})}$,\\
	27: \quad \quad \quad 
	$\eta^2=\frac{n(n-7)\mp (n^2-3n-16)\sqrt{7-n}}{2(2n^3-7n^2-16n+128)}$\\
	28: $\Rightarrow$ $\mathcal{I}(\phi)=w_1\left[\phi(\bm{\nu_1})+\phi(-\bm{\nu_1})\right]+w_2\left[\phi(\bm{\nu_2})+\phi(-\bm{\nu_2})\right]$			\\
	\quad \quad \quad \quad \quad  $+w_3\left[\phi(\bm{\nu_3})+\phi(-\bm{\nu_3})\right]$
}
\vspace{1mm}\\
\textbf{Remark 1.} Each class of cubature rules exhibits different advantages and disadvantages in dealing with various nonlinear systems. This paper compares their performance of approximation accuracy, numerical stability and computational cost when applied to a complex GTE system. This should provide guidelines to be used as benchmark and reference for handling other systems and applications. \\
\textbf{Remark 2.} For tackling FD of GTE system that is considered in this paper the degree of the cubature rule selected  is up to the \textit{5th-degree}. Using higher degrees of cubature rules \textit{do not} yield improved or better solutions for the GTE case study. A significant observation is that a cubature rule over spherical surface that is constituted by the \textit{$d$th-degree} spherical rule and the \textit{$d$th-degree} radial rule can achieve a \textit{$d$th-degree} estimation accuracy. Nevertheless, the degree of the spherical rule as shown in Algorithm 1 is not necessarily equal to the radial rule degree which facilitates and motivates  the \textit{mixture-degree} of cubature rules over the spherical surface \cite{jia2013high,wang2017mixed}.

\subsection{Cubature-based Nonlinear Filters (CNF)}
The proposed CNF in this paper represent as Bayesian filtering approaches that are developed on the basis of the class of $d$th-degree cubature rules. We have utilized the Genz and Mysovskikh theorem-based  cubature rules, and a specially designated but efficient \textit{5th-degree} cubature rule for systems with order $2 \le n \le 7$. The pseudo-code of  procedures for designing the $d$th-degree cubature rules are provided in Algorithm 1 that are derived from  references \cite{genz2003fully,jia2013high,lu2004higher,mysovskikh1980approximation,mysovskikh1981interpolatory,stroudapproximate,arasaratnam2009cubature}.

The appropriate choice of  cubature rules for developing a nonlinear estimation filter essentially relies on the prior knowledge of the system and the user's requirements. First, the degree $d$ that is obtained from the \textit{a priori} knowledge of the system dynamics order (based on Eq. (\ref{eq:dx})) enables one to achieve approximations that are empowered with high accuracy and second, the prior knowledge of the order $n$ enables one to seek a more specific theorem for developing the cubature rules that provide the most  proper achievable approximation given the order range and third, the trade-offs to be made between user's computational efficiency requirements, degree of accuracy, and finally estimation error boundedness guarantees  (since higher $d$ and $n$ result in higher computational complexity and increase in likelihood of negative weights).

\begin{table}[h]
	\caption{The CNF  used for GTE System based on class of cubature rules.}
	\vspace{-3mm}	
	\label{tab:CNBFs}
	\centering
	\footnotesize
	\setlength{\tabcolsep}{1pt}
	\begin{tabular}{ll}
		\toprule
		{CNF} & {Description}\\ \cmidrule(lr){1-1}
		\cmidrule(lr){2-2}
		CNF-I   & \textit{3rd-degree} Genz-spherical rule \& \textit{3rd-degree} radial rule \\
		CNF-II  & \textit{5th-degree} Genz-spherical rule \& \textit{5th-degree} radial rule \\
		CNF-III & \textit{3rd-degree} Mysovskikh-spherical rule \& \textit{3rd-degree} radial rule \\
		CNF-IV  & \textit{5th-degree} Mysovskikh-spherical rule \& \textit{5th-degree} radial rule \\
		CNF-V   & \textit{3rd-degree} Mysovskikh-spherical rule \& \textit{5th-degree} radial rule \\
		CNF-VI   & \textit{5th-degree} Stroud-based cubature rule \\
		\bottomrule
	\end{tabular}   	
\end{table}

The procedure for developing CNF consists of prediction and update steps that are identified in conventional CKF using \textit{3rd-degree} cubature rules, whereas for our schemes   multiple cubature degree rules are employed with associated different projected cubature points and  weights. The six  CNF schemes constructed  for application to GTE are given in Table \ref{tab:CNBFs}. 

\section{Hybrid-Degree Dual Estimation-based Fault Diagnosis Methodology}
\label{sec4}

A novel hybrid-degree dual nonlinear filtering strategy is proposed in this section. To improve the robustness to unmodelled dynamics and uncertainties, a modified cubature point propagation is further incorporated into the hybrid solution. Finally, the fault diagnosis framework is formulated.

\subsection{Hybrid-Degree Dual Estimation Strategy}
Our proposed dual estimation scheme is developed by running two filters concurrently. At every time step, the first CNF-based state filter estimates the states by using the current available estimate of the parameters, $\hat{{\theta}}_{k-1|k-1}$, wheares the second CNF-based parameter filter estimates the health-related parameters by using the current estimate of the states, $\hat{{x}}_{k|k}$.
\\
\textbf{Remark 3.} The key feature  and novelty of our proposed hybrid strategy is that the degree of accuracy for state and parameter estimations are \textit{case-dependent} or determined based on certain performance metrics by using case-dependent cubature rules. Given that the system state process model is highly nonlinear, in general, higher-degree cubature rules are necessary for designing the state estimation filter, whereas for performing parameter estimation lower-degree cubature rules can be acceptable and sufficient. 

In the next subsections, design of concurrent state/parameter estimation filters are provided and the fault diagnosis methodology is introduced and formally specified.

\subsection{$d^x$th-degree Cubature-based State Estimation}
The goal pursued in this subsection is to approximate the objective function in $\mathbb{E}(\phi_1(x_k)|z_{1:k},{\theta}_{k-1})$ that is specified in Eq. (\ref{eq:obj_state}).
The parameter vector ${\theta}_{k-1}$ is assumed to be given and fixed at $\hat{\theta}_{k-1}$ during the state estimation process. Assume that the \textit{$d^x$th-degree} cubature rules are implemented for designing the state estimation filter given the cubature points and weights that are specified and set as $\{\xi_i^{d_x},w_i^{d_x}\}$, $i=1,\cdots,N_{d_x}$.

The class of CNF schemes that are considered for implementing the state estimation filter is introduced next. Given $\phi_1(x_k)=x_k$ and the distribution at time $k-1$ as $\hat{x}_{k-1|k-1} \sim \mathcal{N}(x_{k-1};\hat{x}_{k-1|k-1}, P_{k-1|k-1}^{xx})$, an approximation to predictive $\mathbb{E}(x_k|\theta_{k-1},z_{1:k-1})$ can be first obtained before $z_k$ arrives.
For simplicity, let us denote $\bar{U}_{k-1}=[\hat{\theta}_{k-1}^{T}, u_{k-1}^{T}]^T$. 

{\textit{\textbf{CNF-I}}}:
For the \textit{3rd-degree} nonlinear filter the predictive expectation $\mathbb{E}(x_k|z_{1:k-1},\theta_{k-1})$ is represented by:
\begin{equation}
\footnotesize
\begin{aligned}
&\mathbb{E}(\cdot)\!=w(n_x) \sum_{i=1}^{m(n_x)} \left[{f}\left((\sqrt{\kappa_{I}^x{P_{k-1}^{xx}}}[\textbf{e}]_i+\hat{x}_{k-1}),\bar{U}_{k-1}\right)+{f}\left((-\sqrt{\kappa_{I}^xP_{k-1}^{xx}}[\textbf{e}]_i+\hat{x}_{k-1}), \bar{U}_{k-1}\right)\right]\\
\end{aligned}
\end{equation}
where $w(n_x)={1}/{2{n_x}}$, $m(n_x)=n_x$, $[\textbf{e}]_i$ denotes the unit vector in $\mathbb{R}^{n_x}$ with the $i$th element being 1, and $\kappa_{I}^x=n_x$.  

{\textit{\textbf{CNF-II}}}: For the \textit{5th-degree} nonlinear filter the predictive expectation $\mathbb{E}(x_k|z_{1:k-1},\theta_{k-1})$ is approximated by:
\begin{equation}
\label{eq:CNF2}
\footnotesize
\begin{aligned}
&\mathbb{E}(\cdot)\!=\!w_0(n_x) {f}(\hat{x}_{k-1},\! \bar{U}_{k-1})\!+\!w_1(n_x)\!\!\! \!\sum_{i=1}^{m_1(n_x)}\! \!\left[{f}\left((\sqrt{\kappa_{II}^x {P_{k-1}^{xx}}}[\textbf{s}]_i^{+}\!\!+\!\hat{x}_{k-1}),\bar{U}_{k-1}\right)\!\!+\! {f}\left((\!-\!\sqrt{\kappa_{II}^x {P_{k-1}^{xx}}} [\textbf{s}]_i^{+}\!+\!\hat{x}_{k-1}), \bar{U}_{k-1}\right)\right]\!\! \\
&\!\!+\!w_1(n_x)\sum_{i=1}^{m_2(n_x)}\!\!\! \left[{f}\left((\sqrt{\kappa_{II}^x {P_{k-1}^{xx}}} [\textbf{s}]_i^{-}\!+\!\hat{x}_{k-1}),\! \bar{U}_{k-1}\right)\! +\! f\left((\!-\!\sqrt{\kappa_{II}^x {P_{k-1}^{xx}}} [\textbf{s}]_i^{-}\!+\!\hat{x}_{k-1}),\bar{U}_{k-1}\right)\right]\! \\
&+\!w_2(n_x)  \sum_{i=1}^{m_3(n_x)} \left[{f}\left((\sqrt{\kappa_{II}^x {P_{k-1}^{xx}}} [\textbf{e}]_i+\hat{x}_{k-1}),   \bar{U}_{k-1}\right) +f\left((-\sqrt{\kappa_{II}^x {P_{k-1}^{xx}}} [\textbf{e}]_i+\hat{x}_{k-1}),\bar{U}_{k-1}\right)\right]
\end{aligned}
\end{equation}
where $m_1(n_x)=m_2(n_x)={n_x}({n_x}-1)/2$, and $m_3(n_x)=n_x$. The weights are provided in Table \ref{tab:sum_points_weights_states}.
The points $[\textbf{s}]_i^{+}$ and $[\textbf{s}]_i^{-}$ are defined as $\sqrt{{1}/{2}}([\textbf{e}]_k+[\textbf{e}]_l):k<l,k,l=1,\cdots n_x$ and $\sqrt{{1}/{2}}([\textbf{e}]_k-[\textbf{e}]_l):k<l,k,l=1,\cdots n_x$, respectively,  and $\kappa_{II}^x=n_x+2$ is the scaling factor.

{\textit{\textbf{CNF-III}}}: For the \textit{3rd-degree} nonlinear filter the predictive expectation $\mathbb{E}(x_k|z_{1:k-1},\theta_{k-1})$ is represented by:
\begin{equation}
\label{eq:CNF3}
\footnotesize
\begin{aligned}
&\mathbb{E}(\cdot)=w(n_x) \sum_{i=1}^{m(n_x)} \left[{f}\left((\sqrt{\kappa_{III}^x {P_{k-1}^{xx}}}[\textbf{a}]_i+\hat{x}_{k-1}),\bar{U}_{k-1}\right)+ {f}\left((-\sqrt{\kappa_{III}^x {P_{k-1}^{xx}}}[\textbf{a}]_i+\hat{x}_{k-1}),\bar{U}_{k-1}\right)\right]\\
\end{aligned}
\end{equation}
with $w(n_x)={1}/{2{(n_x+1)}}$, $m(n_x)=n_x+1$, and $\kappa_{III}^x=n_x$.

{\textit{\textbf{CNF-IV}}}: For the \textit{5th-degree} nonlinear filter the predictive expectation $\mathbb{E}(x_k|z_{1:k-1},\theta_{k-1})$ is approximated by:
\begin{equation}
\label{eq:CNF4}
\footnotesize
\begin{aligned}
&\mathbb{E}(\cdot)\!=\!w_0(n_x) {f}\left(\hat{{x}}_{k-1},\bar{U}_{k-1}\right) \!+\!w_1(n_x)\!\!\! \sum_{i=1}^{m_1(n_x)}\!\!\! \left[ f\left((\sqrt{\kappa_{IV}^x {P_{k-1}^{xx}}}\!\!\times\![\textbf{a}]_i\!+\!\hat{x}_{k-1}),\bar{U}_{k-1}\right)\!+\!f\left((-\sqrt{\kappa_{IV}^x {P_{k-1}^{xx}}}[\textbf{a}]_i \!+\!\hat{x}_{k-1}),\bar{U}_{k-1}\right)\right]  \\
& \!+\!w_2(n_x)\sum_{i=1}^{m_2(n_x)}\left[f\left((\sqrt{\kappa_{IV}^x {P_{k-1}^{xx}}}[\textbf{b}]_i+\hat{x}_{k-1}), \bar{U}_{k-1}\right)\!+f\left((-\sqrt{\kappa_{IV}^x{P_{k-1}^{xx}}}[\textbf{b}]_i+\hat{x}_{k-1}),\bar{U}_{k-1}\right)\right]
\end{aligned}
\end{equation}
where $m_1(n_x)=n_x+1$ and $m_2(n_x)=n_x(n_x+1)/2$. The weights are defined in Table \ref{tab:sum_points_weights_states}, and  $\kappa_{IV}^x=n_x+2$.

{\textit{\textbf{CNF-V}}}: For the \textit{mixture-degree} nonlinear filter the predictive expectation $\mathbb{E}(x_k|z_{1:k-1},\theta_{k-1})$ is approximated by:
\begin{equation}
\label{eq:CNF5}
\footnotesize
\begin{aligned}
&\mathbb{E}(\cdot) \!=\!w_0(n_x) {f}(\hat{{x}}_{k-1},\bar{U}_{k-1})\!+\!w_1(n_x)\!\!\sum_{i=1}^{m(n_x)} \!\!\left[ f\left((\sqrt{\kappa_{V}^x {P_{k-1}^{xx}}}[\textbf{a}]_i+\hat{x}_{k-1}),\bar{U}_{k-1} \right)\!+\!f\left((\!-\!\sqrt{\kappa_{V}^x {P_{k-1}^{xx}}}[\textbf{a}]_i\!+\!\hat{x}_{k-1}),\bar{U}_{k-1}\right)\right]\\
\end{aligned}
\end{equation}
where $\kappa_{V}^x=n_x+2$, and the weights are defined in Table \ref{tab:sum_points_weights_states}.

{\textit{\textbf{CNF-VI}}}: For the re-defined \textit{5th-degree} nonlinear filter based on the Stroud's theorem \cite{stroudapproximate} the predictive expectation $\mathbb{E}(x_k|z_{1:k-1},\theta_{k-1})$ is represented by:
\begin{equation}
\label{eq:CNF6}
\footnotesize
\begin{aligned}
&\mathbb{E}(\cdot)=w_1(n_x) \sum_{i=1}^{m_1(n_x)} \left[ f\left((\sqrt{{P_{k-1}^{xx}}}[\bm{\nu_1}]_i+\hat{x}_{k-1}),\bar{U}_{k-1}\right)+f\left((-\sqrt{{P_{k-1}^{xx}}}[\bm{\nu_1}]_i+\hat{x}_{k-1}),\bar{U}_{k-1}\right)\right]+w_2(n_x)\sum_{i=1}^{m_2(n_x)}
\\
&\left[ f\left((\sqrt{{P_{k-1}^{xx}}}[\bm{\nu_2}]_i \!+\!\hat{x}_{k-1}),\bar{U}_{k-1}\right)f\left((\!-\!\sqrt{{P_{k-1}^{xx}}} [\bm{\nu_2}]_i+\hat{x}_{k-1}),+\bar{U}_{k-1} \right)\right]
+w_3(n_x) \sum_{i=1}^{m_3(n_x)}\left[ f\left((\sqrt{{P_{k-1}^{xx}}}[\bm{\nu_3}]_i+\hat{x}_{k-1}),\right.\right.\\
&\left.\left.\bar{U}_{k-1}\right)+f\left((-\sqrt{{P_{k-1}^{xx}}}[\bm{\nu_3}]_i+\hat{x}_{k-1}),\bar{U}_{k-1}\right)\right]
\end{aligned}
\end{equation}
where the weights $w_1(n_x)$, $w_2(n_x)$ and $w_3(n_x)$ are deterministic values with respect to the specific system order. For completeness, the cubature points and weights for all orders can be obtained from \cite{stroudapproximate}, and for $n>7$, some of the cubature points take on complex values.

\begin{table}[!h]
	\caption{Summary of CNF proposed for the state estimation.} \label{tab:sum_points_weights_states}
	\footnotesize
	\vspace{-4mm}
	\renewcommand{\arraystretch}{1.0}
	\setlength{\tabcolsep}{1em} 
	\begin{center}
		\begin{tabular}{l c l l l}
			\toprule
			{Filter}&{{$d^{x}$}}&$\bm{{\xi}}_i^{d_x}$ & $\bm{w}_i^{d_{x}}$ & CR\\
			\cmidrule(lr){1-1} \cmidrule(lr){2-2} \cmidrule(lr){3-4}
			\cmidrule(lr){5-5}
			\multirow{1}{4em}{CNF-I}&{\textit{3rd}}  & $\{[\bm{e}]_i\}$ & $ 1/(2n_x)$ & \cite{genz2003fully,arasaratnam2009cubature}  \\
			\cmidrule(lr){2-2} \cmidrule(lr){3-4}
			\cmidrule(lr){5-5}
			\multirow{4}{4em}{CNF-II}&  & 0  & $ {2}/{(n_x+2)}$    & \\
			&& $\{[\bm{s}]_i^{+}\}$    & $ {1}/{((n_x+2)^2)}$ &  \\
			&{\textit{5th}} & $\{[\bm{s}]_i^{-}\}$    & $ {1}/{((n_x+2)^2)}$ & \cite{jia2013high}
			\\
			&& $\{[\bm{e}]_i\}$  & $ {(4\!-\!n_x)}/{(2(n_x+2)^2)}$  &  \\
			\cmidrule(lr){2-2} \cmidrule(lr){3-4}
			\cmidrule(lr){5-5}
			\multirow{1}{4em}{CNF-III}&{\textit{3rd}}  & $\{[\bm{a}]_i\}$ & $ {1}/{(2(n_x+1))}$ & \cite{mysovskikh1981interpolatory,wang2013spherical}  \\  \cmidrule(lr){2-2} \cmidrule(lr){3-4}
			\cmidrule(lr){5-5}
			\multirow{4}{4em}{CNF-IV}&  & 0  & $ {2}/{(n_x+2)}$  & \\
			&{\textit{5th}}& $\{[\bm{a}]_i\}$   & $ {n_x^2(7\!-\!n_x)}/{(2(n_x\!+\!1)^2 (n_x\!+\!2)^2)}$   & \cite{lu2004higher} \\
			&& $\{[\bm{b}]_i\}$  & $ {2(n_x\!-\!1)^2}/{((n_x\!+\!1)^2 (n_x\!+\!2)^2)}$  &  \\
			\cmidrule(lr){2-2} \cmidrule(lr){3-4}
			\cmidrule(lr){5-5}
			\multirow{2}{4em}{CNF-V}&  & 0  & $ {2}/{(n_x+2)}$  &  \\
			&{\textit{mixture}} & $\{[\bm{a}]_i\}$   & $ {n_x}/{(2(n_x+1) (n_x+2))}$ & \cite{arasaratnam2009cubature,wang2017mixed} \\
			\cmidrule(lr){2-2} \cmidrule(lr){3-4}
			\cmidrule(lr){5-5}
			\multirow{3}{4em}{CNF-VI}&  & $\{\bm{[\nu_1]}\}_i$  &  $w_0 (n_x)$  &  \\
			&{\textit{5th}}& $\{[\bm{\nu_2}]_i\}$   & $ w_1 (n_x) $ & \cite{stroudapproximate}\\
			&& $\{[\bm{\nu_3}]_i\}$   & $ w_2 (n_x) $   &  \\
			\bottomrule
		\end{tabular}
	\end{center}
	Note: ``CR'' denotes the employed cubature rules in the corresponding filters.
\end{table}

Based on $\hat{x}_{k|k-1}=\mathbb{E}(x_k|z_{1:k-1},\theta_{k-1})$ and Table \ref{tab:sum_points_weights_states}, one can further implement the procedure that are identified in conventional CKF using \textit{3rd-degree} cubature rules, whereas for our schemes multiple cubature degree rules are employed with associated different projected cubature points and weights. The differences among the class of CNFs can be identified in terms of the cubature points and weights as determined by the cubature rules that are projected onto the integration domain, leading to different performance on computing the integration of the conditional expectation $\hat{x}_{k|k}=\mathbb{E}(x_k|z_{1:k},\theta_{k-1})$. \\
\textbf{Remark 4.} This subsection explicitly derives a class of CNFs using the compiled cubature rules in Section \ref{sec2}, which enables one in an accessible manner implementation of these methods. The class of CNFs exhibits different advantages and disadvantages in dealing with various nonlinear systems. The performance of the proposed schemes with respect to the approximation accuracy, estimation error boundedness, robustness to unmodelled dynamics and uncertainties, and computational cost when applied to the complex GTE system are quantitatively evaluated and compared in Section \ref{sec5}.

\subsection{$d^\theta$th-degree Cubature-based Parameter Estimation} \label{subsec:par_est}
\subsubsection{Modeling of Parameter Evolution}

In terms of long-term degradation, many works consider an exponential growth with respect to the operating time. For our proposed model-based fault parameter estimation module in this work, we consider simple linear model (Model I) and exponential model (Model II) for the short-term and long-term degradations, respectively.\\
\textbf{\textit{\underline{Model I:}}} For a linear fault or degradation evolution model with uniform time-step the parameter evolution is considered to be governed by
\begin{equation}
\label{eq:par_linear}
\theta_k=\theta_{k-1}+\alpha \Delta t+\tau_k
\end{equation}
where $\alpha$ denotes the growth coefficient and $\Delta t$ denotes a known time-step length.\\
\textbf{\textit{\underline{Model II:}}} For an exponential evolution of the parameter, the model takes the following form
\begin{equation}
\label{eq:par_exp}
\theta_k=e^{\alpha \Delta t} \theta_{k-1}+\beta (1-e^{\alpha \Delta t}) +\tau_k
\end{equation}
where $\alpha$ and $\beta$ denote model coefficients corresponding to the parameter evolution. The coefficient $\beta$ is a scaling factor that can in practice be tuned to better fit the measurement records.

For sake of generality, the state-space model corresponding to  parameters are represented to be  governed by Eq. (\ref{eq:nonlinear_sys_par}),
where $h(\cdot)$ can be  linear as in Model I or exponential as in Model II. The states are assumed to be fixed at $\hat{x}_{k|k}$ that is determined ad specified  from the state estimation filter module.

\subsubsection{CNF for Parameter Estimation}

The main goal here is to approximate the high-dimensional expectation integral $\mathbb{E}(\phi_2(\theta_k)|z_{1:k},{x}_{k})$ given by Eq. (\ref{eq:obj_par}) in the region $x \in \mathbb{R}^{n_\theta}$ on the premise of a Gaussian assumption. In this subsection, it is assumed that state variables $\hat{x}_{k|k}$ are available in order to design our proposed parameter estimation filter. Let us assume that the \textit{$d^\theta$th-degree} cubature rules are implemented for designing the parameter estimation filter, where the cubature points and weights set are given by $\{\xi_j^{d_\theta},w_j^{d_\theta}\}$, $j=1,\cdots,N_{d_\theta}$. The process is referred to as state estimation for obtaining the conditional expectation of $\mathbb{E}(\theta_k|z_{1:k},x_{k})$ and CNF for parameter estimation as summarized in Table \ref{tab:sum_points_weights_par}.
\begin{table}[!h]
	\caption{Summary of CNF proposed for parameter estimation.} \label{tab:sum_points_weights_par}
	\vspace{-4mm}
	\footnotesize
	\renewcommand{\arraystretch}{1.0}
	\setlength{\tabcolsep}{1em} 
	\begin{center}
		\begin{tabular}{l c l l l}
			\hline
			{Filter}&{$d^{\theta}$}&$\bm{{\xi}}_i^{d_\theta}$ & $\bm{w}_i^{d_{\theta}}$& CR \\
			\cmidrule(lr){1-1}
			\cmidrule(lr){2-2} \cmidrule(lr){3-4}
			\cmidrule(lr){5-5}
			CNF-I&\textit{3rd}  & $\{[\bm{e}]_j\}$ & $ 1/2n_\theta$ & \cite{genz2003fully,arasaratnam2009cubature} \\
			\cmidrule(lr){1-1}
			\cmidrule(lr){2-2} \cmidrule(lr){3-4}
			\cmidrule(lr){5-5}
			&  & 0  & $ {2}/{n_\theta+2}$ &  \\
			CNF-III&\textit{3rd}  & $\{[\bm{a}]_j\}$ & $w: {1}/{(2(n_\theta+1))}$ & \cite{mysovskikh1981interpolatory} \\
			\cmidrule(lr){1-1}
			\cmidrule(lr){2-2} \cmidrule(lr){3-4}
			\cmidrule(lr){5-5}
			&  & 0  & $ {2}/{(n_\theta+2)}$    & \\
			CNF-V&\textit{mixture}& $\{[\bm{a}]_j\}$  & $ {n_\theta}/{(2(n_\theta+1) (n_\theta+2))}$& \cite{jia2013high,wang2013spherical}\\
			\hline
		\end{tabular}
	\end{center}
\end{table}

Unlike the state estimation problem, the parameter estimation problem introduced in this subsection is addressed by using the \textit{3rd-degree} or the \textit{mixture-degree} cubature rules. This is justified based on observation that the health parameter dynamics in our engine application are modeled by linear Model I and exponential Model II that are as in general of lower complexity than that of the state dynamics \cite{hanachi2015framework}. The \textit{5th-degree} or higher-degree cubature rules theoretically can be used, however, from the computational efficiency perspective, the higher-degree cubature rules are not recommended as they could lead to substantial computational burden without yielding proportionally improved accuracy and performance.

Let us denote the following error matrices as follows
\begin{equation}
\label{eq:error_matrices}
\begin{split}
\Xi_{k|k-1}^1 &\!\!=\!\!\left[h(\tilde{\xi}_{k-1|k-1,j}^{d_\theta})\!\!-\!\hat{\theta}_{k|k-1},\! \cdots\! h(\tilde{\xi}_{k-1|k-1,N_{d_\theta}}^{d_\theta})\!\!-\!\hat{\theta}_{k|k-1}\right]\!^T\\
\Xi_{k|k-1}^2 &\!\!=\!\!\left[ {\tilde{\xi}}_{k|k-1,1}^{d_\theta}\!-\!\hat{\theta}_{k|k-1},\!\cdots\! {\tilde{\xi}}_{k|k-1,N_{d_\theta}}^{d_\theta} -\hat{\theta}_{k|k-1}\right]\!^T\\
\Xi_{k|k-1}^3 &\!\!=\!\!\left[ g({\tilde{\xi}}_{k|k-1,1}^{d_\theta})\!\!-\!\hat{z}_{k|k-1},\!\cdots\! g{({\tilde{\xi}}_{k|k-1,N_{d_\theta}}^{d_\theta})}\!-\!\hat{z}_{k|k-1}\right]\!^T\\
\end{split}
\end{equation}
\begin{table*}
	\caption{The proposed '$\text{Hybrid} \{i-j\}$' strategy for dual estimation  with $i,j \in \{\text{I,II,III,IV,V,VI}\}$ denoting the CNF as given in Table \ref{tab:CNBFs}.}	
	\label{tab:methods}
	\vspace{-2mm}
	\centering
	\footnotesize
	\renewcommand{\arraystretch}{1}
	\setlength{\tabcolsep}{10pt}
	\begin{tabular}{cllclc}
		\toprule	
		\multirow{2}{4em}{Group}& \multirow{2}{7em}{Methodology}& \multicolumn{2}{c}{State  Filter $\text{CNF}-i$ } & \multicolumn{2}{c}{Parameter Filter $\text{CNF}-j$} \\
		\cmidrule(lr){3-4} \cmidrule(lr){5-6}
		&& $i$ & Degree $d$ & $j$ & Degree $d$ \\ \hline
		\multirow{3}{3em}{G-I}& \multirow{3}{7em}{$\text{Hybrid} \{i-\text{I}\}$}  & $i \in \{\text{I}, \text{III}\}$ & {\textit{3rd-degree}} & \multirow{3}{5em}{$j \in \{\text{I}\}$}&\multirow{3}{6em}{\textit{3rd-degree}}\\
		&  & $i \in \{\text{II}, \text{IV}, \text{VI}\}$ & {\textit{5th-degree}}&&\\
		&  & $i \in \{\text{V}\}$ & {\textit{mixture-degree}}&&\\
		\midrule
		\multirow{3}{3em}{G-II}& \multirow{3}{7em}{$\text{Hybrid} \{i-\text{III}\}$}  & $i \in \{\text{I}, \text{III}\}$ & {\textit{3rd-degree}} & \multirow{3}{5em}{$j \in \{\text{III}\}$}&\multirow{3}{6em}{\textit{3rd-degree}}\\
		&  & $i \in \{\text{II}, \text{IV}, \text{VI}\}$ & {\textit{5th-degree}}&&\\
		&  & $i \in \{\text{V}\}$ & {\textit{mixture-degree}}&&\\
		\midrule
		\multirow{3}{3em}{G-III}& \multirow{3}{7em}{$\text{Hybrid} \{i-\text{V}\}$}  & $i \in \{\text{I}, \text{III}\}$ & {\textit{3rd-degree}} & \multirow{3}{4em}{$j \in \{\text{V}\}$}&\multirow{3}{6em}{\textit{mixture-degree}}\\
		&  & $i \in \{\text{II}, \text{IV}, \text{VI}\}$ & {\textit{5th-degree}}&&\\
		&  & $i \in \{\text{V}\}$ & {\textit{mixture-degree}}&&\\
		\midrule
		\multirow{2}{3em}{G-IV}& \multirow{1}{6em}{Dual-PF}  & PF & - & PF&-\\
		& \multirow{1}{6em}{Dual-UKF}  & UKF & - & UKF&-\\
		\bottomrule
	\end{tabular}   	
\end{table*}
We are now in a position to present the following algorithm.
\\
\textit{{\textbf{Algorithm 2:}} The procedure for our proposed square-root \textit{$d^\theta$th-degree} CNF for  parameter estimation is now introduced.}
\begin{itemize}
	\item[1)] Draw $N_{d_{\theta}}$ cubature points $\xi_j^{d_{\theta}}$ and weights $w_j^{d_\theta}$ based on the $d^{\theta}$th-degree cubature rule and previous distribution $\mathcal{N}(\hat{\theta}_{k-1|k-1},{P}_{k-1|k-1}^{\theta \theta})$, with $j=1, \cdots N_{d_{\theta}}$ and $P_{k-1|k-1}^{\theta \theta}=S_{k-1|k-1}^{\theta \theta} (S_{k-1|k-1}^{\theta \theta})^T$.
	\item[2)] Propagate the sampled cubature points ${\xi}_{j}^{d_\theta}$ as follows\\
	\begin{equation}
	\label{eq:sampling}
	\tilde{\xi}_{k-1|k-1}^{{d_\theta}}=S_{k-1|k-1}^{\theta \theta} {\xi}^{d_\theta}+\hat{\theta}_{k-1|k-1}
	\end{equation}
	\item[3)] Evaluate and predict the states by\\
	$\hat{\theta}_{k|k-1}=\sum_{j=1}^{N_{d_\theta}} w_j^{d_\theta} h\left(\tilde{\xi}_{k-1|k-1,j}^{d_\theta}\right)$ \\
	and obtain the square-root version of the prediction error covariance by\\
	$S_{k|k-1}^{\theta \theta}={qr}([{{\Xi_{k|k-1}^1}}/\sqrt{N_{d_\theta}}\quad S_{\Sigma_{\tau,k-1}}])$.
	\item[4)] Draw and re-propagate the cubature points with the predicted value by\\
	${\tilde{\xi}}_{k|k-1,j}^{d_\theta}=S_{k|k-1}^{\theta \theta} {\xi}_{j}^{d_\theta}+\hat{\theta}_{k|k-1}$.
	\item[5)] Estimate the predicted measurement by evaluating \\
	$\hat{z}_{k|k-1}^{\theta}=\sum_{j=1}^{N_{d_\theta}} w_j^{d_\theta} g\left({\tilde{\xi}}_{k|k-1,j}^{d_\theta},\hat{x}_{k|k},u_k\right)$ \\
	and obtain the square-root version of innovation covariance matrix as\\
	$S_{zz,k|k-1}={qr}([\Xi_{k|k-1}^3/\sqrt{N_{d_\theta}}\quad S_{\Sigma_{v,k}}])$.
	\item[6)] Compute the cross-covariance matrix by firstly consider the square-root version of
	$\mathcal{C}_{k|k-1}= \Xi_{k|k-1}^2//\sqrt{N_{d_\theta}}$, and secondly obtain the cross-covariance matrix as
	$P_{\theta z,k|k-1}=\Xi_{k|k-1}^2 (\Xi_{k|k-1}^3)^{T}/N_{d_\theta}$.
	\item[7)] Update the parameters by invoking	
	$\hat{\theta}_{k|k}=\hat{\theta}_{k|k-1}+K_k^\theta (z_k-\hat{z}_{k|k-1})$, with
	$K_k=P_{\theta z,k|k-1}(P_{zz,k|k-1})^{-1}$.
	The square-root error covariance matrix is now given by\\
	$S_{k|k}^{\theta \theta}={qr} ([\Xi_{k|k-1}^2/\sqrt{N_{d_\theta}}-K_k^\theta\Xi_{k|k-1}^3/\sqrt{N_{d_\theta}} \quad K_k^\theta S_{\Sigma_{v,k}}])$.	
\end{itemize}

\subsection{Modified Cubature Points Propagation}\label{subsec:robust_CP}

This subsection presents a modified cubature point propagation update strategy to enhance the robustness capability of our proposed dual cubature-based scheme to deal with modeling uncertainties. For the parameter estimation module, the $d^{\theta}$th-degree CNF can indeed capture the process dynamics by analyzing the known function $\phi_2(\cdot)$. However, in case of non-negligible uncertainties and unmodelled dynamics in the measurement model, the actual estimation can become compromised by only using the \textit{3rd-degree} cubature rules.

In order to enable cubature points to at least account for both the mean and covariance of the process functions approximate errors, the following conditions are now proposed to be employed in our methodology  as modified cubature point propagation update strategy  \cite{tian2013novel,cui2017improved}, namely consider
\begin{equation}
\label{eq:relation1}
\begin{split}
\Xi_{k|k-1}^1 w^{d_\theta}&=0\\
\Xi_{k|k-1}^1 W^{d_\theta} (\Xi_{k|k-1}^1)^T&=P_{k|k-1}^{\theta\theta}-\Sigma_{\tau,k}\\
\end{split}
\end{equation}
\begin{equation}
\label{eq:relation2}
\begin{split}
\Xi_{k|k}^1 w^{d_\theta}&=0\\
\Xi_{k|k}^1 W^{d_\theta} (\Xi_{k|k}^1)^T&=P_{k|k}^{\theta\theta}-\Delta E\\
\end{split}
\end{equation}
where $\Xi_{k|k}^1=\left[h(\tilde{\xi}_{k|k,1}^{d_\theta})-\hat{\theta}_{k|k}, \cdots, h(\tilde{\xi}_{k|k,N_{d_\theta}}^{d_\theta})-\hat{\theta}_{k|k}\right]^T$, and $W^{d_\theta}=diag([w_1^{d_\theta},\cdots,w_{N_{d_\theta}}^{d_\theta}])$. 	
Let us assume $\Xi_{k|k}^1=\Upsilon_k\Xi_{k|k-1}^1$, and substitute it into Eq. (\ref{eq:relation2}), to obtain $\Xi_{k|k}^1 W^{d_\theta} (\Xi_{k|k}^1)^T=\Upsilon_k(P_{k|k-1}^{\theta\theta}-\Sigma_{\tau,k})\Upsilon_k^T$. Given the following equations
\begin{equation}
\label{eq:Lk_minus}
L_k^{-}(L_k^{-})^T=P_{k|k-1}^{\theta \theta}-\Sigma_{\tau,k-1}
\end{equation}
\begin{equation}
\label{eq:Lk_plus}
L_k^{+}(L_k^{+})^T=P_{k|k}^{\theta \theta}-\Delta E_k
\end{equation}
one can then obtain $\Upsilon_k=L_k^{+}(L_k^{-})^{-1}$. The estimation error covariance is now expressed as
\begin{equation}
\label{eq:PKK}
\begin{split}
P_{k|k}^{\theta \theta}&=(I_{n_\theta}-K_k^\theta  B_k)P_{k|k-1}(I_{n_\theta}-K_k^\theta  B_k)^T+K_k^{\theta}(\mathbb{E}((\zeta_k+\psi_z) (\zeta_k+\psi_z)^T)+\Sigma_{v,k})(K_k^{\theta})^T
\end{split}
\end{equation}
where $\zeta_k$	and $\psi_z$ denote the modeling uncertainty and high order terms resulting from the Taylor series expansion, and $B_k=\partial g(\hat{x}_{k|k},\theta_{k},u_k)/\partial{\theta}$. Therefore, the term $\Delta E_k$ in Eq. (\ref{eq:Lk_plus}) is defined as $\Lambda_k K_k \Sigma_{v,k}K_k^T$, where  $\Lambda_k$ is selected as the largest eigenvalue of $P_{k|k}^{\theta\theta}$ at the time instant $k$. We are now in a position to present our Algorithm 3.

\textbf{Algorithm 3:} The procedure for the modified propagation of cubature points is constructed as follows:
\begin{itemize}
	\item[1)] Generate $\Xi_{k-1|k-1}$ by using $\mathcal{N}(\hat{\theta}_{k-1|k-1},P_{k-1|k-1}^{\theta \theta})$
	\item[2)] Modified cubature points are then generated according to
	\begin{equation}
	\label{eq:sampling_modified}
	\tilde{\xi}_{k-1|k-1}^{d_{\theta}}=\Xi_{k-1|k-1}+\hat{\theta}_{k-1|k-1}
	\end{equation}
	\item[3)] Run the Step 3) of Algorithm 1
	\item[4)] Obtain
	$\Xi_{k|k-1}^1$ from Eq. (\ref{eq:error_matrices}) and obtain $L_k^{-}$ by Eq. (\ref{eq:Lk_minus})
	\item[5)] Run the Steps  5) to 7) of Algorithm 1
	\item[6)] Compute $L_k^{+}$ as in  Eq. (\ref{eq:Lk_plus}) and compute $\Xi_{k|k}$. Then set
	$\tilde{\xi}_{k|k}=\Xi_{k|k}+\hat{\theta}_{k|k}$.
\end{itemize}
\textbf{Remark 5.} By comparing Eq. (\ref{eq:sampling}) with Eq. (\ref{eq:sampling_modified}), it follows that the normal cubature point propagation method depends on the Gaussian assumption of the posterior probability density function (pdf), whereas the modified method relaxes the limitation on the Gaussian
assumption of the posterior pdf.\\
\textbf{Remark 6.} The \textit{5th-degree} cubature rules are significantly more robust to non-Gaussian noise and uncertainties when compared to the \textit{3rd-degree} CKF, UKF and PF \cite{jia2013high}. In presence of measurement uncertainties, the modified cubature points propagation method can be utilized for the parameter estimation scheme given that the \textit{3rd-degree} cubature rules are employed. Due to the fact that this paper concentrates on our proposed dual estimation-based FD methodologies, further comparisons will be implemented on the hybrid-degree solutions. Nevertheless, the robustness analysis with respect to measurement uncertainties and unmodelled dynamics are conducted in  Section \ref{subsec:robust} to demonstrate and illustrate the capabilities and benefits of our accomplished solutions.

\subsection{The Proposed Fault Diagnosis (FD) Formulation}
Diagnosis of drifts in unmeasurable health-related component parameters requires prior knowledge of parameters under healthy condition. Our FD logic and decision making protocol is developed based on analysis of residuals by comparing  estimated parameters  obtained by CNF schemes with parameters that are estimated under the initial fault free operation of the system. The FD problem under consideration deals with the nonlinear system whose dynamics is now governed by
\begin{equation}
\label{eq:FD_system}
\begin{aligned}
x_{k+1}&=f(x_k,\theta_k^T\lambda_\theta(x_k),u_k)+w_k\\
z_k&=g(x_k,\theta_k^T\lambda_\theta(x_k),u_k)+\zeta(x_k,u_k)+v_k
\end{aligned}
\end{equation}
where $\lambda_\theta(x_k)$ denotes the system health parameters representing  a known differentiable function that determines the relationship between the system states and the fault parameters. The component FD problem is tackled and solved by considering that each health parameter is affected by an unknown and time-varying multiplicative fault parameter vector ${\theta}_k$.

Since true values of parameters are assumed  \textit{unknown}, the required residuals for determining the FD criteria are obtained through the so-called \textit{residual signals}. These signals are constructed as the difference between the estimated parameters under the fault-free operational mode (during the very start of the system operation) that is denoted by $\hat{\theta}_h$, and the estimated parameters subsequent to the initial start of the system operation under the possibly faulty  mode that is denoted by $\hat{\theta}_{k|k}$, that is
\begin{equation}
\label{eq:residual}
r_k=\hat{\theta}_h-\hat{\theta}_{k|k}
\end{equation}
where $r_k \in \mathbb{R}^{n_\theta}$. For implementation of our proposed FD strategy developed based on the hybrid-degree dual estimation scheme, the parameter estimates error will be considered as the main indicator for diagnosing faults in the system components. The decision-making logic for detecting, isolating, and identifying the faults are given as follows.

\textit{\textbf{Fault Detection Decision Logic:}} The decision on occurrence of a fault is made when  at least one element of the residual signal in Eq. (\ref{eq:residual}) exceeds its corresponding threshold, i.e.,
if $\forall m, \ \mathbb{E}(\|r_k^m\|)\le r_{max}^m$, the system is classified as healthy; otherwise if $\exists m, \mathbb{E}(\|r_k^m\|)> r_{max}^m$, the system is classified  as being in the faulty condition, where $m \in \{M1,\cdots, M8\}$ denotes the fault mode as explicitly defined in Section V.C.

\textit{\textbf{Fault Isolation Decision Logic:}} The $m$th fault mode is isolated if $\mathbb{E}(\|r_k^m\|) > r_{max}^m$. In case of multi-mode fault scenarios, multiple residuals will exceed their corresponding thresholds.

The variable $r_{max}^m$ denotes the upper bound threshold for the $m$th residual signal as given by Eq. (\ref{eq:residual}). The threshold for each residual signal is selected by conducting Monte Carlo simulation runs using the healthy system I/O data such that missed  alarms and false alarms are minimized corresponding to the healthy mode of the system operation.

\textit{\textbf{Fault Identification Logic:}}
Once the fault modes are detected and isolated, their severity levels through the parameter estimation module are identified based on the magnitude of the residual signals $r_k^m$.

Our proposed hybrid-degree dual estimation strategy features \textit{case-dependent} cubature rules and corresponding CNF for both state estimation and parameter estimation. The details are depicted in Table \ref{tab:methods}, where $i$ denotes the specific CNF for the state estimation module and $j$ refers to the specific parameter estimation module.

Four groups G-I to G-IV of hybrid-degree methodologies are compared and investigated for different purposes. G-I aims to evaluate and compare estimation and fault diagnosis performance under the fixed \textit{3rd-degree} CNF-I for parameter estimation but varying degree of cubature rules for state estimation. G-II and G-III replace the parameter estimation filter in G-I, attempting to evaluate and compare whether different cubature theorems affect the performance of parameter estimation schemes. Dual estimation performance are compared with the well-known PF and UKF that are included in the group G-IV in order to evaluate the accuracy and computational cost of our methodology with these state-of-the-art nonlinear estimation techniques.

\section{Fault Diagnosis of a Gas Turbine Engine (GTE)} \label{sec5}

\subsection{Modeling Overview of Gas Turbine Engines}

The capabilities, advantages, and benefits of our proposed hybrid-degree CNF-based dual estimation strategy are now investigated and demonstrated by applying it to the FD problem of a twin-spool GTE. The FD performance is verified when the GTE is subjected to degradations in its component health parameters by injecting various concurrent/simultaneous abrupt or slowly-varying faults. For a high fidelity representation of the GTE dynamical characteristics the volume dynamics and rotor dynamics are considered, as well as the heat transfer dynamics since they contribute to the nonlinear behavior of the twin-spool GTE \cite{meskin2010fault}. The mathematical model as constructed in \cite{meskin2011multiple} is a set of nonlinear equations of motion that are expressed by Eq. (\ref{eq:states}).

For the physical significance of the model parameters and details refer to \cite{meskin2010fault,meskin2011multiple}. The state variable for the GTE is given by $x=[T_{CC},N_{1},N_{2},P_{LT},P_{CC},P_{LC},P_{HT}]^T$ and the measurement is designated by $z=[N_1,N_2,P_{HC},T_{HC},T_{LC},P_{LC},T_{LT},T_{HT}]^T$, where $T_{CC}$, $T_{HC}$, $T_{LC}$, $T_{LT}$ and $T_{HT}$ represent the temperature variables in combustion chamber (CC), high pressure compressor (HPC), low pressure compressor (LPC), low pressure turbine (LPT) and high pressure turbine (HPT), respectively. $N_1$ and $N_2$ denote the rotational speeds of the spool connecting the HPC to HPT, and the spool connecting the LPC to LPT, respectively. $P_{LT}$, $P_{CC}$, $P_{LC}$, $P_{HT}$, $P_{HC}$, and $P_{LC}$ denote the pressure variables in the subscripted components. The input or the control signal of the twin-spool GTE is the power level angle (PLA) which is related to the fuel mass flow rate ($\dot{m}_f$) through a variable gain. We now have,
\begin{equation}
\label{eq:states}
\resizebox{0.6\textwidth}{!}
{$
	\begin{aligned}
	\dot{T}_{CC}&\!=\!\frac{1}{c_v m_{CC}}[(c_pT_{HC}\theta_{m_{HC}}\dot{m}_{HC\!}+\!\eta_{CC}H_u \dot{m}_f\!-\!c_pT_{CC}\theta_{m_{HT}}\dot{m}_{HT})\\
	&\!-\!c_vT_{CC}(\theta_{m_{HC}}\dot{m}_{HC}\!+\!\dot{m}_f\!-\!\theta_{m_{HT}}\dot{m}_{HT})]\\
	\dot{N}_1   &\!\!=\!\frac{\eta_{mech}^1 \theta_{m_{HT}}\dot{m}_{HT}c_p(T_{CC}\!-\!T_{HT})\!-\!\theta_{m_{HC}}\dot{m}_{HC}c_p(T_{HC}\!-\!T_{LC})}{J_1 N_1 (\frac{\pi}{30})^2}\\
	\dot{N}_2   &\!=\!\frac{\eta_{mech}^2 \theta_{m_{LT}} \dot{m}_{LT}c_p(T_{HT}\!-\!T_{LT})\!-\!\theta_{m_{LC}}\dot{m}_{LC}c_p(T_{LC}-T_{d})}{J_2 N_2 (\frac{\pi}{30})^2}\\
	\dot{P}_{LT}&=\frac{RT_M}{V_M}(\theta_{m_{LT}}\dot{m}_{LT}+\frac{\beta}{1+ \beta}\theta_{m_{LC}}\dot{m}_{LC}-\dot{m}_n)\\
	\dot{P}_{CC}&=\frac{P_{CC}}{T_{CC}}\dot{T}_{CC}+\frac{\gamma R T_{CC}}{V_{CC}}(\theta_{m_{HC}}\dot{m}_{HC}+\dot{m}_f-\theta_{m_{HT}}\dot{m}_{HT})\\
	\dot{P}_{LC}&=\frac{RT_{LC}}{V_{LC}}(\frac{1}{1+\beta}\theta_{m_{LC}}\dot{m}_{LC}-\theta_{m_{HC}}\dot{m}_{HC})\\
	\dot{P}_{HT}&=\frac{RT_{HT}}{V_{HT}}(\theta_{m_{HT}}\dot{m}_{HT}-\theta_{m_{LT}}\dot{m}_{LT})\\
	\end{aligned}
	$}
\end{equation}

During the engine lifetime the compressor and turbine undergo degradations that can originate from various sources, such as fouling, erosion and corrosion that are aerodynamic or performance-related challenges and derivations. These performance-related anomalies can affect the component behavior and eventually the overall behavior of the GTE system. Component faults that are of concern in this paper are caused by fouling and erosion degradations since they contribute to significant deterioration in the engine life cycle \cite{pourbabaee2015sensor}.

Fouling always occur in compressors (both at the HPC and LPC segments), which cause changes in compressor mass flow rate and efficiency. Erosion phenomena exert effects on reduction of efficiency and increase of  mass flow rate in  HPT or LPT segments. Consequently, the health parameters that are considered in this paper relate to efficiency and mass flow rates in compressor  caused by fouling, as well as in turbine  caused by erosion. A fault vector  $$[\theta_{\eta_{LC}},\theta_{\eta_{HC}},\theta_{\eta_{LT}},\theta_{\eta_{HT}},\theta_{\dot{m}_{LC}},\theta_{\dot{m}_{HC}},\theta_{\dot{m}_{LT}},\theta_{\dot{m}_{HT}}]^T$$ is incorporated into the mathematical model (\ref{eq:states}) to manifest impacts of health parameters in corresponding components. The subscript $\eta$ implies the change of efficiency and the subscript ${\dot{m}}$ implies the change of mass flow rate.

\subsection{Verification and Validation of the Model Subject to Uncertainties}
To verify and validate the effectiveness of our proposed  strategy the design of our nonlinear filters is based on a \textit{simplified} mathematical model as provided in Eq. (\ref{eq:states}), however all the simulations shown subsequently have been applied to a more detailed, complex, and accurate model of the GTE that is obtained from GSP10 \cite{meskin2010fault,meskin2011multiple,tsoutsanis2014component}.

The differences between the simplified model (\ref{eq:states}) and the high fidelity representation of the GTE that is obtained from  GSP10 \cite{meskin2010fault,meskin2011multiple,tsoutsanis2014component} capture uncertainties and unmodeled dynamics. These  are attributed  to the manner performance maps are constructed that can express  relationships between the health parameters and the system states as denoted by    $\zeta(x_k,u_k)$ in Eq. (\ref{eq:FD_system}). Specifically, the performance maps for efficiencies and mass flow rates of the compressors (including both the HPC and LPC segments)  that correspond to $\dot{m}_{HC}$, ${\eta}_{HC}$, $\dot{m}_{LC}$ and  ${\eta}_{LC}$ in the model (\ref{eq:states}), as well as the performance maps for efficiencies and mass flow rates of the turbines (including both the HPT and LPT segments) that correspond to $\dot{m}_{HT}$, ${\eta}_{HT}$, $\dot{m}_{LT}$ and  ${\eta}_{LT}$ in the model (\ref{eq:states}) need to be estimated and identified.

Performance maps  used in the GTE thermodynamic model are generated through various methodologies in the literature, such as \cite{tsoutsanis2014component}. In this paper, the methodology that is used for generating performance maps for the compressors and turbines is through twelve multi-layer feed-forward neural networks. The networks are used for identifying the relationships between the concerned health parameters and the pressure ratio, as well as the states.

\begin{figure*}
	\centering
	\resizebox{1\columnwidth}{!}{\includegraphics{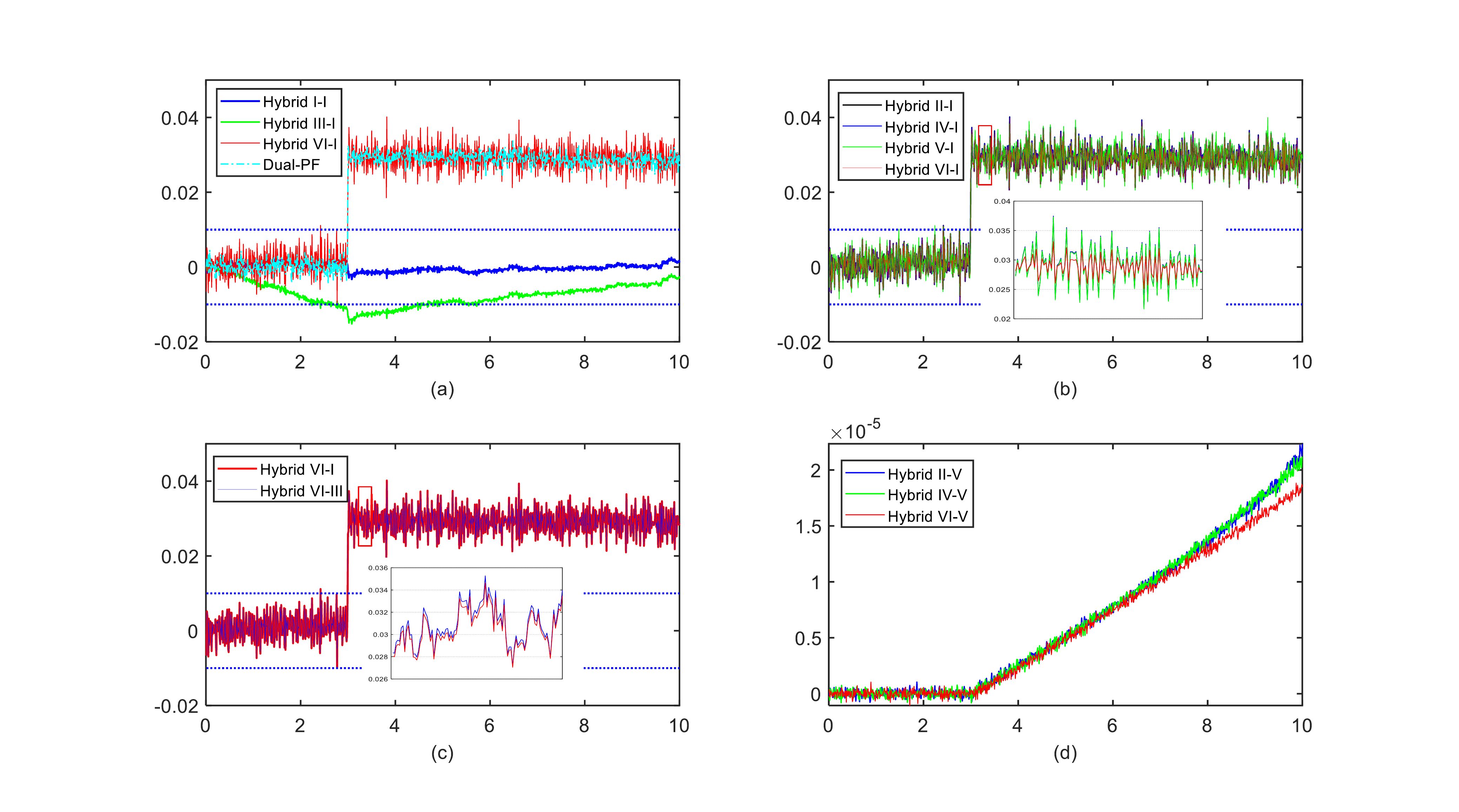}}\vspace{-6mm}	
	\caption{{Residuals $r_k$ ($y$-axis) corresponding to the abrupt 3\% mass flow rate fault scenario in the HPC injected at the time $t=3s$.} }
	\label{fig:foul_fmHC}
\end{figure*}

An extensive set of simulation studies are conducted to ensure that the simplified model used in Eq. (\ref{eq:states}) is sufficiently reliable with respect to the more detailed, accurate, and high fidelity model of the GTE for further conducting our case studies robustness to uncertainties and unmodelled dynamics. These details are provided in Section \ref{subsec:robust}.

\subsection{Hybrid-Degree Fault Diagnosis Performance Analysis}
The goal of this subsection is to justify and verify the rationalization and effectiveness of our proposed hybrid-degree dual CNF schemes through simulations  under various fault scenarios. All the simulation scenarios correspond to the cruise flight mode of the GTE, and the process and measurement noise levels correspond to the same values as provided  in \cite{meskin2011multiple}, where  standard deviations are given as percentage of the nominal values at typical cruise operating conditions. The PLA is assumed to be at 0.9, the Mach number is 0.74, and the ambient conditions are set to standard conditions. Importantly, since our goal is to compare capabilities of our proposed nonlinear filters it is justifiable that all comparative studies associated with the considered methodologies are implemented on the basis of the same process and measurement noise distributions for both state and parameter estimation problems.

Our main objective is focused on FD performance  of the GTE system. The hybrid-degree combinations are provided in Table \ref{tab:methods}, where the Dual-UKF and Dual-PF, and Hybrid {I-I} are effectively three concurrently running UKF, two PF and two CKF. For implementing the Dual-PF, the number of particles is selected through a quantitative analysis that is derived based on the mean absolute error (MAE\%) accuracy criterion with respect to the estimation process steady state values. The number of particles is chosen as $500$ corresponding to both the state and parameter estimation filters for the GTE. The number of cubature points for the CNF and the unscented points for the UKF are deterministic values and are provided in Table \ref{tab:no_points}. 
Below we provide details on our considered three (3) distinct case studies where the fault modes are explicitly defined in Table \ref{tab:fault_mode}:

\begin{table}[!th]
	\caption{Degradation Modes Considered in the Gas Turbine Engine.}
	\vspace{-4mm}
	\label{tab:fault_mode}
	\footnotesize
	\renewcommand{\arraystretch}{0.9}
	\setlength{\tabcolsep}{10pt}	
	\begin{center}
		\begin{tabular}{cllc}
			\toprule
			Component & HP & Description  & Mode\\
			\cmidrule(lr){1-1} \cmidrule(lr){2-3} \cmidrule(lr){4-4}
			\multirow{2}{3em}{HPC}&\multirow{1}{3em}{${\theta}_{\eta_{HC}}$} & Changes in efficiency of HPC&M1\\
			&\multirow{1}{2em}{${\theta}_{\dot{m}_{HC}}$} & Changes in mass flow rate of HPC&M2\\
			\multirow{2}{3em}{HPT}&\multirow{1}{3em}{${\theta}_{\eta_{HT}}$} & Changes in efficiency of HPT&M3\\
			&\multirow{1}{2em}{${\theta}_{\dot{m}_{HT}}$} & Changes in mass flow rate of HPT&M4\\
			\multirow{2}{3em}{LPC}&\multirow{1}{3em}{${\theta}_{\eta_{LC}}$} & Changes in efficiency of LPC&M5\\
			&\multirow{1}{2em}{${\theta}_{\dot{m}_{LC}}$} & Changes in mass flow rate of LPC&M6\\
			\multirow{2}{3em}{LPT}&\multirow{1}{3em}{${\theta}_{\eta_{LT}}$} & Changes in efficiency of LPT&M7\\
			&\multirow{1}{2em}{${\theta}_{\dot{m}_{LT}}$} & Changes in mass flow rate of LPT&M8\\
			\bottomrule
		\end{tabular}
	\end{center}
\end{table}

\begin{table*}
	\caption{State/parameter estimation accuracy (MAE\%) for Case I corresponding to abrupt faults in the HPC.}	 
	\label{tab:estimation_results1}
	\centering
	\footnotesize
	\renewcommand{\arraystretch}{1.0}
	\setlength{\tabcolsep}{1.6pt}	
	\begin{tabular}{ccccccccccccccc}
		\toprule
		\multirow{3}{2em}{Var.}&\multicolumn{2}{c}{{Hybrid \{I-I\}}} &\multicolumn{2}{c}{Hybrid \{II-I\}} &\multicolumn{2}{c}{Hybrid \{IV-I\}}  &\multicolumn{2}{c}{{Hybrid \{VI-I\}}} &\multicolumn{2}{c}{{Dual-PF}} &\multicolumn{2}{c}{Hybrid \{II-III\}} &\multicolumn{2}{c}{Hybrid \{VI-III\}} \\
		\cmidrule(lr){2-3} \cmidrule(lr){4-5} \cmidrule(lr){6-7} \cmidrule(lr){8-9}  \cmidrule(lr){10-11} \cmidrule(lr){12-13} \cmidrule(lr){14-15}
		& {B-F} & {C-I} & {B-F} & {C-I}  & {B-F} & {C-I}   & {B-F} & {C-I} 	& {B-F} & {C-I} & {B-F} & {C-I}  & {B-F} & {C-I}   \\ \midrule
		$TCC$ &1.2000  &1.1200 &0.2000  &0.1200 &0.2004 &0.2028   &0.2003 &0.2027 &0.1903  &0.1310  &0.2010  &0.1211 &0.2011 &0.2035\\
		$N1$  &16.446  &18.857 &3.6905  &3.2481 &4.6905 &3.7379   &4.6904 &3.7379 &3.6805  &3.2235  &3.6902  &3.2485 &4.6911 &3.7392\\	
		$N2$  &5.4174  &6.3313 &1.4090  &1.2491 &1.2864 &1.2469   &1.2864 &1.2469 &1.2345  &1.1043  &1.4088  &1.2498 &1.2875 &1.2512\\	
		$PLT$ &3.0034  &3.0023 &0.0034  &0.0011 &0.0017 &0.0011   &0.0016 &0.0011 &0.0032  &0.0012  &0.0036  &0.0012 &0.0019 &0.0101\\	
		$PCC$ &1.0632  &1.0387 &0.0623  &0.0076 &0.0113 &0.0076   &0.0113 &0.0075 &0.0701  &0.0083  &0.0627  &0.0077 &0.0121 &0.0127\\	
		$PLC$ &3.0219  &3.0131 &0.0214  &0.0036 &0.0054 &0.0035   &0.0054 &0.0034 &0.0026  &0.0123  &0.0232  &0.0042 &0.0062 &0.0066\\	
		$PHT$ &4.0061  &3.4341 &0.0062  &0.0016 &0.0022 &0.0015   &0.0021 &0.0015 &0.0015 &0.0026   &0.0064  &0.0018 &0.0029 &0.0035\\
		$\theta_{m_{HC}}$ &1.0109 &3.1212 &0.0011 &0.0010 &0.0018 &0.0016 &0.0012 &0.0011 &0.0009   &0.0011  &0.0011 &0.0011 &0.0012 &0.0012\\				
		\bottomrule
	\end{tabular}	
	\begin{tablenotes}
		\item[a]  Note: B-F and C-I denote 'Before-Fault' and 'after Case-I', respectively.
	\end{tablenotes}	
\end{table*}

\subsubsection{\textbf{Case I: Abrupt Faults in the HPC}}
In this scenario,  effects of abrupt faults are studied by injecting a $3\%$ mass flow rate loss (representing the fault severity) affecting the HPC component at the instant $t = 3 s$. The residual signals with respect to the mass flow rate in the HPC are shown in Fig. \ref{fig:foul_fmHC} corresponding to the groups G-I to G-IV. The blue dotted lines depict the confidence bounds for residuals that are determined based on $50$ independent Monte Carlo simulation runs under various healthy scenarios. By analyzing the residuals, the fault can be clearly detected and diagnosed.

Fig. \ref{fig:foul_fmHC} (a) and Fig. 2 (b) depict the comparative results with respect to group G-I where they share the same filter for the parameter estimation (CNF-I), and group G-IV that involves Dual-PF and Dual-UKF. It follows from these results that residuals corresponding to our proposed combinations of the ``\textit{5th-degree cubature rules for state estimation and 3rd-degree for parameter estimation}", including Hybrid \{II-I\}, Hybrid \{IV-I\} and Hybrid \{VI-I\}, as well as the Dual-PF schemes can detect  changes after the fault occurrence and converge to the injected fault severity. However, the Dual-CKF (i.e., Hybrid \{I-I\}) fails to detect the fault occurrence and the Hybrid \{III-I\} shows both false positive and false negative alarms during the indicated time window. What is in common for Hybrid \{I-I\} and Hybrid \{III-I\} methodologies is that both are using the \textit{3rd-degree} cubature rule for designing the state estimation nonlinear filter.  In this case, the \textit{3rd-degree} cubature rules are \textit{not appropriate/suitable} for designing the state estimator under the given noise levels.

Moreover, the Dual-UKF scheme is not capable of detecting the fault, whereas the Dual-PF scheme performs well in terms of detection and residual change tracking. The differences among the \textit{5th-degree} filters and the Dual-PF are not visibly distinguished, therefore the quantified MAE\% is provided in Table \ref{tab:estimation_results1}. Observations from this table indicate that the approximation accuracy of the \textit{5th-degree} filters are quite close to that of the Dual-PF, where the combination of Hybrid \{II-I\} and Hybrid \{VI-I\} are slightly more accurate than others.

Fig. \ref{fig:foul_fmHC} (c) compares the residuals corresponding to the \textit{3rd-degree} Mysovskikh-based CNF (CNF-III) for the parameter estimation purpose. The goal is to analyze whether the theorems affect the performance of dual estimation schemes and  generated residuals. Effectively,  modification of the \textit{3rd-degree} cubature rule for the parameter estimation does not provide obvious influence on the resulting residuals. The quantified estimation accuracy through MAE\% is provided in Table \ref{tab:estimation_results1} by comparing the Hybrid \{VI-I\} with the Hybrid \{VI-III\}, and the Hybrid \{II-I\} with the Hybrid \{II-III\}. The observation from this table is that the \textit{3rd-degree} cubature rule for parameter estimation based on the Genz's theorem yields relatively a higher precision than that is based on the Mysovskikh's theorem for our GTE application.

Fig. \ref{fig:foul_fmHC} (d) aims to analyze the performance of the \textit{mixture-degree} filter for parameter estimation. It follows that the hybrid combination constituted by a \textit{mixture-degree} filter (CNF-V) for parameter estimation cannot detect the fault and converge to the fault severity within the selected time window. \\
\textbf{Remark 7.} The \textit{3rd-degree} cubature rules (i.e., CR-I and CR-III) obtain poor engine state approximation of the statistical moments. Therefore, the corresponding CNF significantly deteriorates the dual estimation and fault diagnosis results (refer to e.g., Hybrid {I-I} or Hybrid {III-I}). Consequently, the importance of a a vlaid method for approximating the statistical moments is demonstrated and emphasized.
\begin{figure*}
	\centering
	\resizebox{1\columnwidth}{!}{\includegraphics{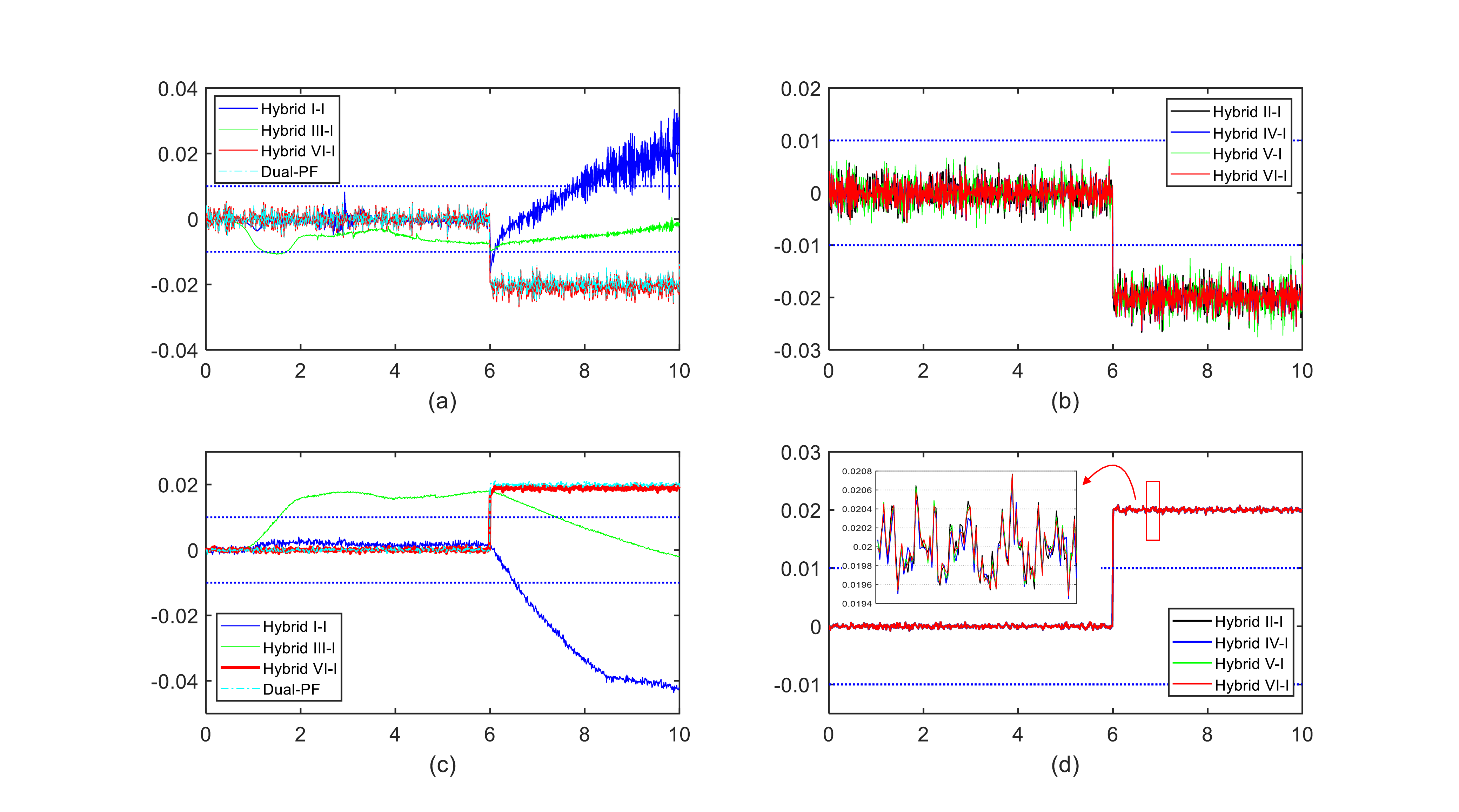}}
	\vspace{-4mm}
	\caption{Residuals $r_k$ ($y$-axis) (\%) for (a) mass flow rate; (b) mass flow rate; (c) efficiency; and (d) efficiency,   simultaneous fault scenario: Case II. }
	\label{fig:foul_fmHT_feHT}
\end{figure*}
\subsubsection{\textbf{Case II: Simultaneous Abrupt Faults in the HPT}}
The effects of simultaneous abrupt faults are investigated by injecting a $2\%$ mass flow rate increase and a $2\%$ efficiency decrease affecting the HPT segment at the instant $t = 6 s$. The residual signals resulting from the class of hybrid-degree combinations for the HPT mass flow rate and efficiency are shown in Fig. \ref{fig:foul_fmHT_feHT}. Results for G-II, G-III and G-IV are not provided since these methodologies \textit{cannot detect faults}, and  methodologies in G-II are not listed as well since their performance have \textit{not been improved}.

The observations from Fig. \ref{fig:foul_fmHT_feHT} can be summarized as follows: \textit{(i)} The \textit{3rd-degree} cubature rules used for both state and parameter estimation \textit{cannot} achieve the FD objectives, where the Hybrid \{I-I\} yields false alarms and \textit{cannot} converge to the correct fault severity. The Hybrid \{III-I\} is \textit{not} capable of detecting the fault occurrence in the mass flow rate (Fig. \ref{fig:foul_fmHT_feHT} (a)) and yields both false positive and negative alarms in efficiency (Fig. \ref{fig:foul_fmHT_feHT} (b)). \textit{(ii)} The proposed hybrid combinations, i.e. Hybrid \{II-I\} and Hybrid \{IV-I\}  \textit{can detect} the fault immediately after its occurrence and do ultimately converge to the \textit{correct} fault severity. Besides, the Dual-UKF scheme \textit{cannot} react to the fault occurrence, while the Dual-PF scheme achieves accurate estimation that are close to our proposed hybrid combinations. Quantitative estimation accuracy results for the Case II using the MAE\% metric is shown in Table \ref{tab:estimation_results2}.\\

\noindent
\textbf{Discussions on FD Performance for Abrupt Fault Cases:}

The purpose of this subsection is to provide comparison on the FD performance of all the methodologies provided in Table \ref{tab:CNBFs} before proceeding to more case studies. The metrics for evaluating the reliability of FD schemes consist of the estimation accuracy,  computational cost and  numerical stability factor (SF). The estimation accuracy is quantitatively measured through MAE\% corresponding to the last two seconds of  simulations after convergence of the filters. The computational cost is evaluated as the number of points or particles, and the numerical stability factor is quantified by $\text{SF}={\sum_{i=1}^{N_d}|w_i^d|}/{\int_{\mathbb{R}^n}{w_d(x)dx}}={\sum_{i=1}^{N_d}|w_i^d|}/{\sum_{i=1}^{N_d}w_i^d}$.

The metric SF manifests the numerical stability capability of the cubature rules, where $\text{SF}=1$ denotes an optimal value, since it implies that the cubature rule holds the weights all-positive. The estimation accuracy based on MAE\% is provided in Table \ref{tab:estimation_results1} and Table \ref{tab:estimation_results2} for the two fault cases. The computational cost as judged by the number of points/particles and SF values for various methods are shown in Table \ref{tab:no_points}. Comparisons lead to the following observations and conclusions:
\begin{itemize}
	\item In view of estimation accuracy, our proposed hybrid schemes with combination of \textit{\textit{5th-degree} CNF for the state estimation and the \textit{3rd-degree} CNF for the parameter estimation} can reach a high accuracy level with respect to MAE\% using the Dual-PF method. The  downside of the other hybrid choices are that some parameter estimates cannot converge to the actual fault severity  and they provide a large number of false alarms (combinations based on the \textit{3rd-degree} for both the state and parameter estimations), and some cannot even detect the fault after its occurrence (combinations based on the \textit{5th-degree} or \textit{mixture-degree} for parameter estimation). The specific theorems affect the performance slightly among the \textit{5th-degree} cubature rules, but generally the accuracy improves significantly by using the \textit{3rd-degree} cubature rules for the state estimation module.
	\item From the perspective of computational cost, to achieve the expected estimation accuracy using the Dual-PF one should employ $500$ particles to perform either state estimation or parameter estimation. This is by far higher than our proposed hybrid-degree schemes. Particularly, the Hybrid \{VI-I\} is the most computationally efficient combination within the approaches that can simultaneously detect, isolate and identify the faults (Table \ref{tab:no_points}).
	\item In view of numerical stability, the \textit{5th-degree} cubature rules (either the Genz's  or  Mysovskikh's theorems) risk of having higher probabilities of instabilities, although the Mysovskikh's theorem is more robust since the negative weights occur when the system order is greater than 7, while  Genz's theorem experiences negative weights when $n \ge 4$. The UKF that is utilized in our GTE application suffer from higher risk of numerical instability for both state and parameter estimation scheme (as shown in Table \ref{tab:no_points}). Importantly, our proposed efficient \textit{5th-degree} rule based on the Stroud's theorem and the CNF-VI filter maintain positive weights for our GTE system which enables them to guarantee their numerical stability.
\end{itemize}
In the following discussions on FD capabilities for the GTE system, we concentrate on comparisons and evaluations of Hybrid \{I-I\} (two concurrently running CKF), Hybrid \{VI-I\} and Dual-PF schemes.
The simulation scenarios consist of multi-mode \underline{concurrent} fault cases and \underline{simultaneous}   fouling  and erosion degradation scenarios. For the fault parameter estimation module dealing with the compressors fouling degradation, the linear model in Eq. (\ref{eq:par_linear}) is selected, whereas, the exponential model in Eq. (\ref{eq:par_exp}) is utilized for the long-term turbine degradation prediction.

\subsubsection{\textbf{Case III: Multi-Mode Concurrent Faults}}
Effects of concurrent faults are investigated by injecting sequential fault patterns into the GTE system first at time $t = 30 s$ where the mass flow rate and efficiency in the LPC segment simultaneously decrease by 3\%; second at the time instance $t = 80 s$ the mass flow rate in the LPT is increased by 2\% and the efficiency is decreased by 2\%; third the HPC segment experiences a 1\% mass flow rate loss and a 4\% efficiency loss at the time instant $t = 120 s$; and finally at $t = 160 s$, the mass flow rate in  HPT is increased by 2\% and the efficiency is decreased by 2\%.

The resulting residual signals are shown in Fig. \ref{fig:concurrent_all}, where the Dual-UKF scheme is not shown since it \underline{cannot} detect fault occurrences in this case. It can be observed that our proposed Hybrid \{VI-I\} demonstrates the best performance as compared to the Dual-PF and Hybrid \{I-I\}, since it can detect and isolate  multi-mode faults at instances of  fault occurrences. Moreover,  estimated fault severities converge to their corresponding true injected fault values. Although the Dual-PF can also achieve the FD objectives for majority of generated residuals, however for the mass flow rate change in HPT at  $t = 160 s$ it fails to converge to the expected 2\% mass flow rate increase. The Hybrid \{I-I\} generates false negative in LPC mass flow rate fault, and the convergence rate of estimated parameters is much slower than the other two methodologies. Therefore,  residuals cannot converge to actual fault severities in the selected time windows.

\begin{table*}
	\caption{State/parameter estimation accuracy (MAE\%) for Case II corresponding to simultaneous abrupt faults in the HPT.}	
	\label{tab:estimation_results2}
	\centering
	\footnotesize
	\renewcommand{\arraystretch}{1.0}
	\setlength{\tabcolsep}{4pt}	
	\begin{tabular}{ccccccccccccc}
		\toprule
		\multirow{3}{2em}{Var.}&\multicolumn{2}{c}{{Hybrid \{I-I\}}} &\multicolumn{2}{c}{Hybrid \{II-I\}} &\multicolumn{2}{c}{Hybrid \{IV-I\}}  &\multicolumn{2}{c}{{Hybrid \{VI-I\}}} &\multicolumn{2}{c}{{Dual-PF}}  &\multicolumn{2}{c}{Hybrid \{VI-III\}} \\
		\cmidrule(lr){2-3} \cmidrule(lr){4-5} \cmidrule(lr){6-7} \cmidrule(lr){8-9}  \cmidrule(lr){10-11} \cmidrule(lr){12-13}
		& {B-F} & {C-II} & {B-F} & {C-II}  & {B-F} & {C-II}  & {B-F} & {C-II}	& {B-F} & {C-II} & {B-F} & {C-II}     \\ \midrule
		$TCC$ &2.1091  &1.9800 &0.2020  &0.1200 &0.3014 &0.2028   &0.2003 &0.2027 &0.1903  &0.1340   &0.2011 &0.3012\\
		$N1$  &2.4461  &20.342 &2.9605  &3.2841 &4.7205 &3.7379   &4.6904 &3.7349 &3.6205  &3.5735   &4.6911 &3.9345\\	
		$N2$  &1.2454  &7.2343 &2.3578  &1.2491 &2.9874 &1.2469   &2.2864 &2.2469 &2.2345  &2.1043   &1.2875 &2.8433\\	
		$PLT$ &0.9244  &3.0623 &0.0026  &0.0011 &0.0017 &0.0011   &0.0016 &0.0011 &0.0036  &0.0025   &0.0019 &0.0111\\	
		$PCC$ &1.0632  &4.0284 &0.0795  &0.0623 &0.0113 &0.0076   &0.0133 &0.0055 &0.0611  &0.0067   &0.0121 &0.0184\\	
		$PLC$ &0.9219  &3.9305 &0.0197  &0.0086 &0.0054 &0.0035   &0.0058 &0.0034 &0.0025  &0.0094   &0.0062 &0.0093\\	
		$PHT$ &1.0061  &3.2463 &0.0102  &0.0099 &0.0022 &0.0015   &0.0021 &0.0015 &0.0015 &0.0026    &0.0029 &0.0124\\
		$\theta_{m_{HT}}$ &0.0009 &2.3241 &0.0010 &0.0012 &0.0020 &0.0016 &0.0012 &0.0011 &0.0008   &0.0012   &0.0010 &0.0009\\		
		$\theta_{\eta_{HT}}$ &0.0102 &4.4216 &0.0014 &0.0012 &0.0020 &0.0016 &0.0012 &0.0011 &0.0009   &0.0013   &0.0010 &0.0126\\				
		\bottomrule
	\end{tabular}	
	\begin{tablenotes}
		\item[a]  Note: B-F and C-II denote 'Before-Fault' and 'after Case-II', respectively.
	\end{tablenotes}	
\end{table*}

Moreover, a given parameter change can cause other parameter changes and slightly affect  estimation of other fault severities. For instance, in  Fig. \ref{fig:concurrent_all} (h),  residuals within the expected fault-free time window (before $t = 160 s$) change slightly when the other faults occur, but  residuals do not exceed their thresholds. This behavior can be explained as a result of uncertainties and discrepancies between the actual  engine model and the simplified mathematical model that was used for the filter design, as well as unavoidable coupling effects among  components that generate the actual engine data. Nevertheless, our proposed methodologies are still able to detect, isolate and identify  fault scenarios and their severities.
\begin{figure*}
	\centering
	\resizebox{1.0\columnwidth}{!}{\includegraphics{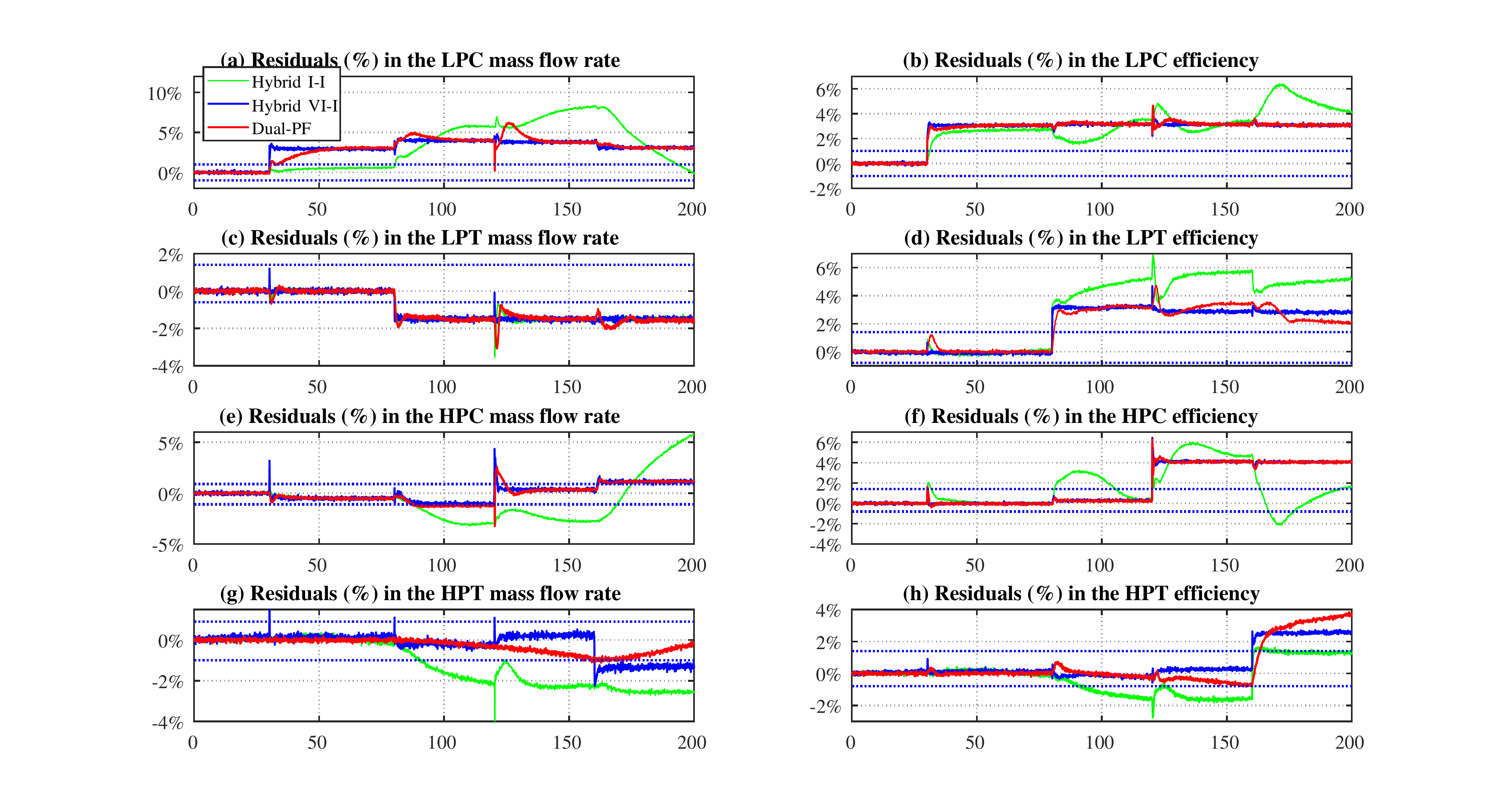}}
	\vspace{-8mm}
	\caption{Residuals $r_k$ (the $y$-axis) (\%) corresponding to  concurrent and simultaneous fault scenarios in the component subsystems LPC, LPT, HPC and HPT. }
	\label{fig:concurrent_all}
\end{figure*}

\begin{table*}[!h]
	\caption{Confusion matrix analysis (\%).} \label{tab:confusion_matrix_1}
	\vspace{-4mm}
	\footnotesize
	\setlength{\tabcolsep}{0.8em} 
	\begin{center}
		\begin{tabular}{ccccccccccc}
			\toprule
			& $ACC$ & $FP$ & $P_{\theta_{mHC}}$ & $P_{\theta_{\eta_{HC}}}$ & $P_{\theta_{mLC}}$ & $P_{\theta_{\eta_{LC}}}$ & $P_{\theta_{mHT}}$ & $P_{\theta_{\eta_{HT}}}$ & $P_{\theta_{mLT}}$ & $P_{\theta_{\eta_{LT}}}$  \\
			\cmidrule(lr){1-1}
			\cmidrule(lr){2-11}
			Hybrid \{I-I\}     & 59.67 & 33.40	 &65.22 &69.57  &62,96 &76.47  &60.00 &47,62 &47.83&58.33\\
			Hybrid \{VI-I\}    & 86.42 & 8.61	 &92.86 &90.00  &93.33 &88.46  &92.59 &80.77 &77.77&67.74\\						 
			Dual-PF           & 87.01 & 7.687	 &91.25 &88.37  &94.56 &89.87  &91.74 &88.93 &77.85&67.83\\
			\bottomrule
		\end{tabular}	
	\end{center}
\end{table*}

\subsection{Fault Diagnosis Comparative Results}
In this subsection, a quantitative study is conducted by utilizing the confusion matrix analysis to evaluate the reliability, accuracy, precision, false alarm and/or misclassification rates corresponding to  methodologies that are proposed in this work.
For each  algorithm (i.e., Hybrid \{I-I\}, Hybrid \{VI-I\} and Dual-PF), the confusion matrices are obtained by performing 100 independent Monte Carlo simulation runs.
Fault scenarios are generated by considering severities that range from $1\%$ to $10\%$ of loss of effectiveness.

\begin{table*}
	\caption{Computational cost with respect to number of points and stability factor (SF) for dual estimation methodologies.}	
	\vspace{-2mm}
	\label{tab:no_points}
	\centering
	\footnotesize
	\renewcommand{\arraystretch}{1.2}
	\setlength{\tabcolsep}{1pt}	
	\begin{tabular}{ccccccccccc}
		\hline 
		&&Hy \{I-I\}&Hy \{II-I\}&Hy \{III-I\}&Hy \{IV-I\}&Hy \{VI-I\}&Hy \{II-III\}&Hy \{VI-III\}&Dual-PF & Dual-UKF\\\hline
		\multirow{2}{5em}{State Estimation}&Ref&$2n_x$&$2n_x^2+1$&$2n_x+2$&$n_x^2+3n_x+3$&$n_x^2+n_x+2$&$2n_x^2+1$&$n_x^2+n_x+2$&-&$2n_x+1$\\
		&GTE &14&99&16&73&58&99&73&500&15\\
		&SF &1&1.23&1&1&1&1&1&-&3.67\\
		\multirow{2}{5em}{Parameter Estimation} &Ref &$2n_x$&$2n_x$&$2n_x$&$2n_x$&$2n_x$&$2n_x+2$&$2n_x+2$&-&$2n_\theta+1$\\
		&GTE &16&16&16&16&16&18&18&500&17\\
		&SF &1&1&1&1&1&1&1&-&4.33\\
		\hline
	\end{tabular}	
	\begin{tablenotes}
		\item[a]  Note: ``Hy'' indicates Hybrid, and ``Ref'' indicates the general reference number of points; 'GTE' indicates the specific number of points for the gas turbine engine.
	\end{tablenotes}
\end{table*}

The rows in  confusion matrices show the actual number of fault scenarios  applied to the GTE system and the columns represent the number of estimated fault categories. The diagonal elements represent the true positive rate ($TP$) for each fault occurrence. The evaluation metrics of the \underline{accuracy}  ($ACC=\sum_{j=1}^{9}c_{jj}/(\sum_{i=1}^{9} \sum_{j=1}^{9} c_{ij})$), \underline{precision} ($P_j=c_{jj}/\sum_{i=1}^{9}c_{ij}$) and \underline{false positive} ($FP$) ($FP=\sum_{j=1}^{8} c_{9j}/\sum_{j=1}^{9}c_{9j}$) are also provided in Table \ref{tab:confusion_matrix_1},
where $c_{ij}$ with $i,j=1,\cdots 9$ denote the value of rows and columns of the confusion matrix.

The results are summarized in Table \ref{tab:confusion_matrix_1} which demonstrate that the FD accuracy of our Hybrid \{VI-I\}  estimation ($86.42\%$) outperforms that of  the Hybrid \{I-I\} approach ($59.67\%$), and the false positive alarm rate of our proposed method ($8.61\%$) is much lower than that of the Hybrid \{I-I\} method ($33.40\%$). The precision of our scheme for all the eight fault parameters is higher than that of Hybrid \{I-I\} approach. The performance of Dual-PF scheme is quite close to that of our designed Hybrid \{VI-I\} approach in terms of ACC, TP and precision, however the former approach needs a much higher computational cost  to achieve the same performance estimation levels.

\subsection{Robustness Analysis in  Presence of Uncertainties}\label{subsec:robust}
The purpose of this section is to evaluate  robustness of the designed hybrid-degree dual cubature-based nonlinear filtering schemes with respect to parametric uncertainties and unmodelled dynamics that arise from modeling. To verify the robustness of our proposed FD framework, the following uncertainties are first considered:
\begin{equation}
\nonumber
\small
\begin{split}
&\zeta_1(x_k,u_k)={(k_1(T_{CC}-T_{HT})-k_2(T_{HC}-T_{LC}))}/({N_2 (\pi/30)^2}\Delta J_1)\\
&\zeta_2(x_k,u_k)={(k_3(T_{HT}-T_{LT})-k_4(T_{LC}-T_{d}))}/({N_1 (\pi/30)^2}\Delta J_2)\\
&\zeta_3(x_k,u_k)=\Delta \gamma RT_{CC}(\dot{m}_{HC}+\dot{m}_{f}-\dot{m}_{HT})/V_{CC}\\
\end{split}
\end{equation}
where $k_1=\eta_{mech}^1\dot{m}_{HT}c_p$, $k_2=\dot{m}_{HC}c_p$, $k_3=\eta_{mech}^2\dot{m}_{LT}c_p$, and $k_4=\dot{m}_{LC}c_p$. The parameters $\Delta \gamma$,  $\Delta J_{1}$ and  $\Delta J_{2}$ correspond to inaccuracies in the ratio,  the inertia of the high and low spool shafts, respectively.
Therefore, the modeling uncertainty is represented by\\ $\zeta(x_k,u_k)=\left[\zeta_1(x_k,u_k),\zeta_2(x_k,u_k),\zeta_3(x_k,u_k),0,0,0,0,0\right]^T$

At time $t = 6 s$, the fault to the LPT component is injected with a  2\% increase in the mass flow rate. The parametric uncertainties of $\Delta \gamma$,  $\Delta J_{1}$ and  $\Delta J_{2}$ in $\zeta(\hat{x}_k,u_k)$ are first assumed to be present at time $t=7s$ with an error of 3\%. The FD performance in terms of  residuals in
presence of  the above modeling uncertainties are shown in Fig. \ref{fig:robust}. It can be observed that our proposed Hybrid \{VI-I\}, Dual-PF and Hybrid \{VI-I\} with modified cubature points propagation can still detect, isolate and identify the faults having different levels of fluctuations, where the Hybrid \{VI-I\} with modified cubature points propagation method can be more robust to  uncertainties. Several false alarms have occurred by using the Dual-PF. In contrast, higher false alarms are generated by the considered Dual-CKF and Dual-UKF.

Table \ref{tab:robustness} shows the robustness analysis when a 6\% increase of the LPT mass flow rate is injected in  presence of different levels of uncertainties. It follows from this table that the Hybrid \{VI-I\} with modified cubature points propagation exhibits the lowest false alarm rates and the best accuracy in terms of $MAE(\%)$, whereas the fault detection time is longer than the Hybrid \{VI-I\} and Dual-PF. As compared to Dual-PF, Dual-CKF, and Dual-UKF schemes, the Hybrid \{VI-I\} scheme can detect occurrence of a fault most quickly and shows a more robust capabilities with respect to  false alarm rates. The Dual-CKF and Dual-UKF are more sensitive to parametric  uncertainties. One can observe that if level of uncertainties is increased e.g. to 7\% inaccuracy, then all methodologies produce false alarms. Therefore, this testing case study scenario can be regarded as a reference benchmark on  limits of our proposed strategy when handling significant levels of simultaneous severe faults and modelling uncertainties.

\begin{figure}[t!]
	\centering
	\resizebox{0.5\columnwidth}{!}{\includegraphics{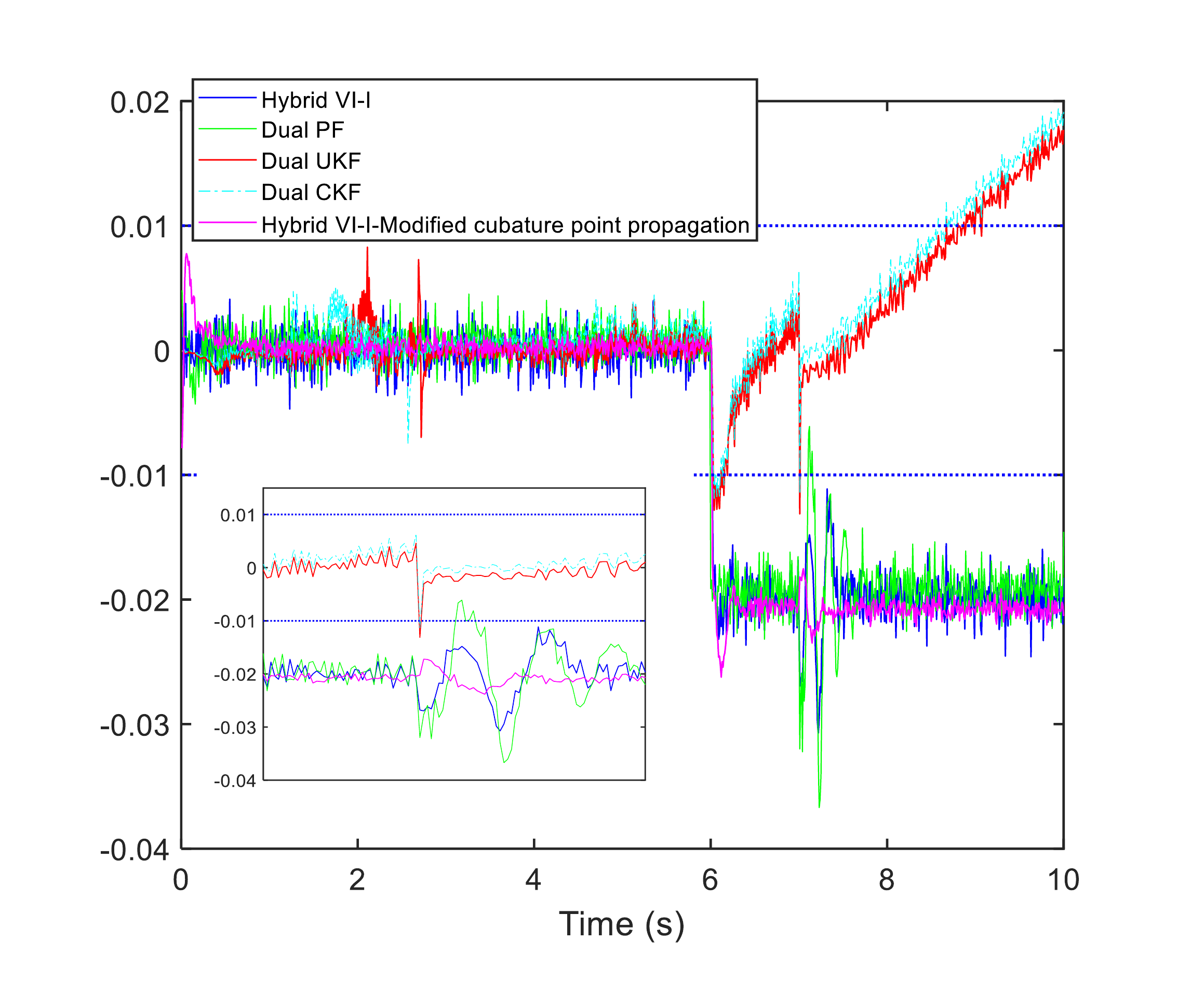}}
	\vspace{-6mm}
	\caption{{Residuals $r_{k}$ (the $y$-axis) (\%) corresponding to  degradation in the LPT mass flow rate scenarios in  presence of modeling uncertainties.} }
	\label{fig:robust}
\end{figure}

\begin{table*}
	\caption{{Robustness analysis for 6\% increase in the LPT mass flow rate corresponding to various uncertainty levels.}}	
	\vspace{-2mm}
	\label{tab:robustness}
	\centering
	\footnotesize
	\renewcommand{\arraystretch}{1.0}
	\setlength{\tabcolsep}{1.5pt}	
	\begin{tabular}{cccccccccccccc}
		\toprule
		\multirow{3}{2em}{Method}&\multirow{3}{2em}{FDT}
		&\multicolumn{2}{c}{{2\%}} &\multicolumn{2}{c}{3\%} &\multicolumn{2}{c}{4\%}  &\multicolumn{2}{c}{{5\%}} &\multicolumn{2}{c}{{6\%}}  &\multicolumn{2}{c}{7\%} \\ \cmidrule(lr){3-4} \cmidrule(lr){5-6} \cmidrule(lr){7-8} \cmidrule(lr){9-10}  \cmidrule(lr){11-12} \cmidrule(lr){13-14}
		&& {MAE\%} & {FAR\%}  & {MAE\%} & {FAR\%} & {MAE\%} & {FAR} & {MAE\%} & {FAR\%} & {MAE\%} & {FAR\%} & {MAE\%} & {FAR\%}    \\ \midrule
		Hybrid \{VI-I\} &0.35&0.08&0&0.12&0&0.19&0.01&0.21&0.03&0.35&0.08&0.51&0.11\\
		Dual-PF  &0.50&0.09&0&0.11&0.01&0.18&0.01&0.22&0.03&0.32&0.07&0.50&0.10\\	
		Dual-CKF  &0.85&0.31&0.02&0.61&0.03&0.78&0.04&1.59&0.07&1.79&0.12&2.01&0.22\\	
		Dual-UKF &0.80&0.33&0.02&0.65&0.03&0.81&0.05&1.78&0.08&2.03&0.26&2.24&0.30\\	
		M-Hybrid \{VI-I\} &0.65&0.05&0&0.07&0&0.10&0&0.14&0&0.18&0&0.22&0.03\\						
		\bottomrule
	\end{tabular}	
	\begin{tablenotes}
		\item[a]  Note: FDT denotes 'Fault Detection Time', and FAR denote 'False Alarm Rate'.
	\end{tablenotes}	
\end{table*}

In terms of the robustness capability against  unmodelled dynamics, additive nonlinearities are added to the measurement model with respect to the spool speed $N_1$. It was shown that (figures  not shown due to space limitations) if the magnitude of uncertainty exceeds beyond a certain range our proposed FD framework could produce erroneous decisions. Nevertheless, it generates significantly improved FD performance as compared to  available methodologies in the literature that we have considered in this work, and it enables one to deal with unknown dynamics within a given bounded range.

\section{Boundedness Analysis of Parameter Estimation Error}\label{sec6}

The following lemmas are essential in establishing our main technical analysis and results.

\begin{lem}[\cite{theodor1996robust}]
	\label{lemma1}
	For $0 \leq k \leq N$, suppose that $X = X^T \geq 0$, $\mathcal{S}_k(X)=\mathcal{S}_k^T(X) \in \mathbb{R}^{n \times n}$ and $\mathcal{H}_k(X)=\mathcal{H}_k^T (X) \in \mathbb{R}^{n \times n}$. If
	\begin{equation}
	\label{eq:lemma11}
	\mathcal{S}_k(Y) \geq \mathcal{S}_k(X), \ \forall X \leq Y = Y^T
	\end{equation}
	and
	\begin{equation}
	\label{eq:lemma12}
	\mathcal{H}_k(Y) \geq \mathcal{S}_k(X)
	\end{equation}
	Then the solutions $\mathcal{M}_k$ and $\mathcal{N}_k$ to the following difference equations
	\begin{equation}
	\label{eq:lemma13}
	\mathcal{M}_{k+1} = \mathcal{S}_k(\mathcal{M}_{k}), \  \mathcal{N}_{k+1} = \mathcal{H}_k(\mathcal{N}_{k}), \ \mathcal{M}_{0} = \mathcal{N}_0 > 0
	\end{equation}
	satisfy
	\begin{equation}
	\label{eq:lemma14}
	\mathcal{M}_{k} \leq \mathcal{N}_k
	\end{equation}
\end{lem}

\begin{lem}[\cite{xie1994robust}]
	\label{lemma2}
	Given matrices $A$, $H$, $E$ and $F$ with appropriate dimensions such that $F F^T \leq I$. Let $X$ be a symmetric positive definite matrix and $\gamma$ be an arbitrary positive constant such that $\gamma^{-1}I-EXE^T>0$. Then the following inequality holds
	\begin{equation}
	\label{eq:lemma21}
	\begin{split}
	(A+HFE)&X(A+HFE)^T \\
	&\leq A(X^{-1}-\gamma E^T E)^{-1} A^T+\gamma^{-1}HH^T
	\end{split}
	\end{equation}
\end{lem}

\begin{lem}[\cite{kluge2010stochastic}]
	\label{lemma3}
	If both $A$ and $B$ are symmetric positive definite matrices, then
	\begin{equation}
	\label{eq:lemma31}
	(A+B)^{-1}>A^{-1}-A^{-1}BA^{-1}
	\end{equation}
\end{lem}

This paper concentrates on the boundedness analysis of parameter estimation error which is of utmost importance to accomplish the fault diagnosis objective, rather than the joint convergence analysis. The conventional strategy for dual state/parameter estimation scheme is to  first optimize one with the other one fixed, and then alternate. Different from the other direct decoupling approaches, the error-coupling effects between states and parameters will be considered for the boundedness analysis of parameter estimation error in the next subsection.

Before proceeding to the boundedness analysis, the following assumption is made regarding the dynamical system (\ref{eq:nonlinear_sys_state}) and (\ref{eq:nonlinear_sys_measure}).
\\
\textbf{Assumption 2:} The variable $\{x_k, \theta_k\}$ satisfies the range over a compact set, for which the functions $f(x_k, \theta_k, u_k)$ in (\ref{eq:nonlinear_sys_state}) and $g(x_k, \theta_k, u_k)$ in  (\ref{eq:nonlinear_sys_measure}) are continuously differentiable with respect to the state $x_k$ and the parameter $\theta_k$, respectively.

The parameter estimation methodology that is developed in this work is based on  \textit{3rd-degree} cubature rules. Our goal is to investigate boundedness properties of the estimated parameters in presence of both approximation errors of the cubature rules, modeling uncertainties and error from the state estimation. Let us consider the following reformulated system ${\Omega}_\theta$ for the parameter estimation problem
\begin{equation}
\label{eq:nonlinear_sys_par_1}
{\Omega}_\theta:  
\left\{ 
\begin{array}{l}
\theta_k=h(\theta_{k-1})+\tau_{k-1}\\
z_k={g}(\hat{x}_{k|k},\theta_k,u_k)+\tilde{g}_{x,k}+\zeta(\hat{x}_{k|k},u_k)+v_k\\
\end{array}
\right.
\end{equation}
where $\tilde{g}_{x,k} \triangleq {g}(x_{k},{\theta}_{k},u_{k})-{g}(\hat{x}_{k|k},{\theta}_{k},u_{k})$ denotes the nonlinear interactive error which is introduced to account for the bias of state estimate $\hat{x}_{k|k}$. 

\subsection{Bounded Parameter Estimation Error Covariance}
The goal here is to verify the boundedness of the parameter estimation error covariance. 
\\
\begin{thm}
	\label{thm:Pkk_theta}
	Consider the nonlinear system (\ref{eq:nonlinear_sys_par_1}), and let the following conditions hold:\\
	(1) There exist positive constants $b_{min}$, $b_{max}$, $\delta_{v,max}$, such that the following bounds are satisfied for $k\geq 0$:
	\begin{equation}
	\label{eq:thm:Pkk_theta1}
	\begin{split}
	&b_{min}^2I \le B_kB_k^T  \le b_{max}^2I, \  \Sigma_{v,k} \le \delta_{v,max}I
	\end{split}	
	\end{equation}
	(2) Taking the high-order terms of the Taylor series expansion into consideration, there exist positive constants, $\gamma_{min}$, $\gamma_{max}$, $d_{i,min}$, $d_{i,max}$, $l_{min}$, $l_{max}$, such that the following bounds can be fulfilled:
	\begin{equation}
	\label{eq:thm:Pkk_theta2}
	\begin{split}
	&\gamma_{min}^2I \le  \Gamma_{k}\Gamma_{k}^T  \le \gamma_{max}^2I\\
	&d_{i,min}^2I \le  \mathcal{D}_{i,k}\mathcal{D}_{i,k}^T  \le d_{i,max}^2I, i=1,2\\
	&l_{min}^2I \le  \mathcal{L}_k\mathcal{L}_k^T  \le l_{max}^2I,
	\end{split}	
	\end{equation}	
	Then, the parameter estimation error covariance matrix can be bounded by
	\begin{equation}
	\label{eq:thm:Pkk_theta3}
	P_{k|k}^{\theta \theta}\leq \lambda_{max}^{\theta} I
	\end{equation}
\end{thm}	

\begin{proof}
	\label{proof:thm:Pkk_theta}
	
	Based on the Taylor series expansion, the parameter and measurement prediction errors by using the proposed CNFs for parameter estimation can be obtained as
	\begin{equation}
	\label{eq:PE_k|k-1_simple}
	\epsilon_{k|k-1}\!=\!A_{k-1}\epsilon_{k-1|k-1}\!+\!\psi_\theta(\hat{\theta}_{k-1|k-1},\theta_{k-1},x_k)\!+\!\tau_{k}
	\end{equation}
		\begin{equation}
	\label{eq:PE_k|k-1}
	\begin{split}
	\epsilon_{k|k}&=(I-K_k^\theta B_k)\epsilon_{k|k-1}+K_k^\theta\psi_{z,k}^\theta (\hat{\theta}_{k|k-1},\theta_k, x_k) \\
	&-K_k^\theta \zeta_k-K_k^{\theta}\tilde{g}_{x,k}-K_{k}^\theta v_k\\
	\end{split}
	\end{equation}
	where $\psi_\theta(\hat{\theta}_{k-1|k-1},\theta_{k-1},x_k)$ and $\psi_z(\hat{\theta}_{k|k-1},\theta_k,x_k)$ represent the higher order terms which involve truncation errors associated with the approximation, and $A_{k-1}=\partial h(\cdot)/\partial \theta_{k-1}$, $B_k=\partial g(\cdot)/\partial \theta_k$.
	
	For simplicity, the high-order terms $\psi_\theta(\hat{\theta}_{k-1|k-1},\theta_{k-1}, x_k)$, $\psi_z(\hat{\theta}_{k|k-1},\theta_k, x_k)$ and $\zeta(x_k, u_k)$ in the following deviations are simplified as $\psi_{\theta,k-1}$, $\psi_{z,k}^{\theta}$ and $\zeta_k$, respectively.
	
	In order to facilitate the expression in the process of boundedness analysis, the conditions on the interactive error term $\tilde{g}_{x,k}$, the high-order terms $\psi_{\theta,k-1}$ and $\psi_{z,k}^{\theta}$ are first analyzed. The condition on uncertainty $\zeta_k$ has been provided in Assumption 1.
	\begin{itemize}
		\item[$\bullet$] The interactive error term $\tilde{g}_{x,k}$ can be bounded given the following assumption.
		\\
		\textbf{Assumption 3:} The state estimation error and its corresponding error covariance matrix at the time instant $k$ are bounded by $\varepsilon_x$ and  $\varsigma_{max} I$, respectively, with $\varepsilon_x>0$ and $\varsigma_{max}>0$.
		
		The rationality of the Assumption 3 will be discussed in Section \ref{subsec:bound_state}. 
	\end{itemize}
	
	Based on the Taylor series expansion of $\tilde{g}_{x,k}$ at $\hat{{x}}_{k|k}$, one obtains $\tilde{g}_{x,k}=\alpha_{g_x}G_{x,k}\tilde{x}_{k|k}$, where $G_{x,k}=\partial g(x_k,\theta_k,u_k)/\partial x_k$ and $\alpha_{g_x}=diag(\alpha_{g_x,1,k},\alpha_{g_x,2,k},\cdots,\alpha_{g_x,n_x,k})$ denotes an unknown instrumental diagonal matrix to compensate the high-order terms of expansion. Given Assumption 2, the two terms $G_{x,k}$ and  $\alpha_{g_x}$ are assumed to hold conditions of $\|G_{x,k}\| \leq \bar{g}_x$ and  $\|\alpha_{g_x}\| \leq \bar{\alpha}_{g_x}I$, respectively. In this case, further considering Assumption 3, one can obtain inequalities $\|\tilde{g}_{x,k}\|\leq \bar{\tilde{g}}_x \varepsilon_{x}$ and  $\|\tilde{g}_{x,k}\tilde{g}_{x,k}^T\|\leq\bar{\tilde{g}}_x^2 \varsigma_{max} I$, with $\bar{\tilde{g}}_x=\bar{\alpha}_{g_x} \bar{g}_x$.
	
	\begin{itemize}
		\item[$\bullet$] The high-order terms $\psi_{\theta,k-1}$ and $\psi_{z,k}^{\theta}$ can be transformed into the following formulations \cite{kai2009robust}
		$$\psi_{\theta,k-1}=\Theta_{k-1}\mathcal{D}_{1,k-1}\mathcal{L}_{k-1}\epsilon_{k-1|k-1}$$
		$$\psi_{z,k}^{\theta}=\Gamma_k \mathcal{D}_{2,k}\mathcal{L}_k \epsilon_{k|k-1}$$ 
		where $\Theta_{k-1}$ and $\Gamma_k$ denote problem-dependent scaling matrices, $\mathcal{L}_{k}$ is introduced to provide an extra degree of freedom to tune the filter, and $\mathcal{D}_{i,k}, i=1,2$  denotes an unknown time-varying matrix accounting for the linearization errors of dynamical model which satisfies $\mathcal{D}_{i,k} \mathcal{D}_{i, k}^T \leq I$. The conditions on these matrices are given in (\ref{eq:thm:Pkk_theta2}).
	\end{itemize}

	Considering the parameter error covariance matrices $P_{k|k-1}^{\theta \theta}=\mathbb{E}\{\epsilon_{k|k-1}\epsilon_{k|k-1}^T\}$ and $P_{k|k}^{\theta \theta}=\mathbb{E}\{\epsilon_{k|k}\epsilon_{k|k}^T\}$, with the Gaussian-assumed update procedures, the error covariance can be approximated by
	\begin{equation}
	\label{eq:P_kk_theta}
	\begin{split}
	&P_{k|k}^{\theta \theta}=(A_{k-1}+\Theta_{k-1}\mathcal{D}_{1,k-1}\mathcal{L}_{k-1})P_{k-1|k-1}^{\theta \theta}(A_{k-1}+\Theta_{k-1}\mathcal{D}_{2,k-1}\mathcal{L}_{k-1})^T-\left\{(A_{k-1}+\Theta_{k-1}\mathcal{D}_{1,k-1}\mathcal{L}_{k-1})P_{k-1|k-1}^{\theta \theta} \right.\\
	&\left.\times(A_{k-1}+\Theta_{k-1}\mathcal{D}_{2,k-1}	\mathcal{L}_{k-1})^T+\Sigma_{\tau,k-1}\right\}(B_{k}+\Gamma_k \mathcal{D}_{2,k}\mathcal{L}_k)^T 
	\left[(B_{k}+\Gamma_k \mathcal{D}_{2,k}\mathcal{L}_k)\left((A_{k-1}+\Theta_{k-1}\mathcal{D}_{1,k-1}\mathcal{L}_{k-1})\right.\right.\\
	&\left.\left.\times P_{k-1|k-1}^{\theta \theta}(A_{k-1}+\Theta_{k-1}\mathcal{D}_{1,k-1}\mathcal{L}_{k-1})^T+\Sigma_{\tau,k}\right)(B_{k}+\Gamma_k \mathcal{D}_{2,k}\mathcal{L}_k)^T+ \mathbb{E}\{\zeta_k\zeta_k^T\}+\mathbb{E}\{\tilde{g}_{x,k}\tilde{g}_{x,k}^T\}+\Sigma_{v_k}
	\right]^{-1}(B_{k}\\
	& 
	+\Gamma_k \mathcal{D}_{2,k}\mathcal{L}_k)\times \left\{(A_{k-1}+\Theta_{k-1}\mathcal{D}_{1,k-1}\mathcal{L}_{k-1})P_{k-1|k-1}^{\theta \theta}(A_{k-1}+\Theta_{k-1}\mathcal{D}_{1,k-1}\times \mathcal{L}_{k-1})^T+\Sigma_{\tau,k}\right\}
	\end{split}
	\end{equation}
	where the term $\tilde{g}_{x,k}$ is uncorrelated with the modeling uncertainty $\zeta_k$ and the predictive error $\epsilon_{k|k-1}$.  
	According to Lemma \ref{lemma3}, one can approximate
	\begin{equation}
	\label{eq:P_kk_theta1}
	\begin{split}
	&P_{k|k}^{\theta \theta}\leq (B_{k}+\Gamma_k \mathcal{D}_{2,k}\mathcal{L}_k)^{-1}\left[\mathbb{E}\{\zeta_k\zeta_k^T\}+\mathbb{E}\{\tilde{g}_{x,k}\tilde{g}_{x,k}^T\}+\Sigma_{v_k}\right](B_{k}+\Gamma_k \mathcal{D}_{2,k}\mathcal{L}_k)^{-T}
	\end{split}
	\end{equation}
	Given the conditions in (\ref{eq:thm:Pkk_theta1}) and (\ref{eq:thm:Pkk_theta2}), we can have
	\begin{equation}
	\label{eq:P_kk_theta2}
	\begin{split}
	&P_{k|k}^{\theta \theta}\leq
	\frac{\bar{\zeta}^2I+\bar{\tilde{g}}_x^2 \varsigma_{max}I+\delta_{v,max}I}
	{(b_{min}+\gamma_{min}d_{2,min}l_{min})^2} 
	\end{split}
	\end{equation}
	Computing the Euclidean norm on both sides of the above inequality leads to:
	\begin{equation}
	\label{eq:P_kk_theta3}
	\begin{split}
	&\|P_{k|k}^{\theta \theta}\|\leq
	{(\bar{\zeta}^2+\bar{\tilde{g}}_x^2 \varsigma_{max}+\delta_{v,max})}/
	{(b_{min}+\gamma_{min}d_{2,min}l_{min})^2} 
	\end{split}
	\end{equation}
	Therefore, the parameter estimation error covariance $P_{k|k}^{\theta \theta}$ is bounded in the case that the modeling uncertainty $\zeta_k$, the interactive error term $\mathbb{E}\{\tilde{g}_{x,k}\tilde{g}_{x,k}^T\}$ and the noise term $\Sigma_{v,k}$ are bounded. Let the right hand side of the inequality (\ref{eq:P_kk_theta2}) be denoted by $\lambda_{max}^{\theta}I$, then the proof of Theorem \ref{thm:Pkk_theta} is completed.
	
\end{proof}

\subsection{Bounded Parameter Estimation Error}

The second task for parameter estimation error boundedness analysis is to provide a sufficient condition to verify the exponential boundedness of the parameter estimation error in the mean square sense. The following Assumption 4 states some standard results on  boundedness of  stochastic processes that are utilized as presented in our main result in Theorem \ref{thm:stability}. \\
\textbf{\text{Assumption 4:}}
It is assumed that \textit{(a)} the matrix $A$ satisfies $\|A_kA_k^T\| \le a_{max}^2I$, \textit{(b)} the prior error covariance matrix satisfies ${\lambda}_{min}^{\theta} I \leq P_{k-1|k-1}^{\theta \theta} \leq {\lambda}_{max}^{\theta} I$, \textit{(c)} the high-order terms-related matrix is bounded by $\|\Theta_{k-1}\Theta_{k-1}^T\| \le \vartheta_{max}^2I $, \textit{(d)} the inequalities $ \Sigma_{\tau,k} \le \delta_{\tau,max}I$ hold, and where all the bounds are positive constants.
\begin{thm}
	\label{thm:stability}
	Consider the parameter estimation filter as proposed in the dual methodology consisting of  a \textit{$3rd$-degree} ($d_\theta=3$) CNF (CNF-I or CNF-III), and let Theorem \ref{thm:Pkk_theta} and Assumption 4 hold. The parameter estimation error $\epsilon_{k|k}$ satisfies the following conditions,
	\begin{equation}
	\label{eq:thm_eq1}
	\mathbb{E}[{J}(\epsilon_{k|k})]-{J}(\epsilon_{k-1|k-1}) \le \mu_q-{\mu_p} {J}(\epsilon_{k-1|k-1})
	\end{equation}
	\begin{equation}
	\label{eq:thm_eq2}
	\frac{1}{{\lambda}_{max}^{\theta}}\|\epsilon_{k-1|k-1}\|^2 \le {J}(\epsilon_{k-1|k-1}) \le \frac{1}{{\lambda}_{min}^{\theta}} \|\epsilon_{k-1|k-1}\|^2
	\end{equation}
	Therefore, $\epsilon_{k|k}$ is boundeded  in mean square sense where
	\begin{equation}
	\label{eq:thm_eq3}
	\mathbb{E}\{\| \epsilon_{k|k}\|^2\} \le \frac{{\lambda^{\theta}_{max}}}{{\lambda}^{\theta}_{min}} \mathbb{E}\{\|\epsilon_{0|0}\|^2\}(1-\mu_p)^{k}+\frac{\mu_q}{{\lambda}^{\theta}_{min}} \sum_{i=1}^{k-1} (1-\mu_p)^i
	\end{equation}
	if the initial conditions of the system satisfy $\|\epsilon_{0|0}\| \le \epsilon_f$, and $\mu_q>0$, $0 < \mu_p \le 1$.
\end{thm}
\begin{proof}
	\label{proof:thm1}
	Let us define a performance index for parameter estimation as	${J}(\epsilon_{k|k})=\epsilon_{k|k}^T (P_{k|k}^{\theta \theta})^{-1} \epsilon_{k|k}$.  Following the Assumption 3, it gives
	\begin{equation}
	\nonumber
	\frac{1}{{\lambda}_{max}^{\theta}}\|\epsilon_{k-1|k-1}\|^2 \le {J}(\epsilon_{k-1|k-1}) \le \frac{1}{{\lambda}_{min}^{\theta}} \|\epsilon_{k-1|k-1}\|^2
	\end{equation}
	Substituting Eq. (\ref{eq:PE_k|k-1_simple}) into $\hat{\theta}_{k|k}$ and Eq. (\ref{eq:PE_k|k-1}), one obtains
	\begin{equation}
	\label{eq:e_k+1|k}
	\begin{split}
	\epsilon_{k|k}&=(A_{k-1}+\Theta_{k-1}\mathcal{D}_{1,k-1}\mathcal{L}_{k-1})(I-K_{k}^\theta (B_{k}+\Gamma_k \mathcal{D}_{2,k}\mathcal{L}_k))\epsilon_{k-1|k-1}+o_{u,k}+o_{n,k}\\
	\end{split}
	\end{equation}
	where $o_{n,k}=(I-K_{k}^\theta (B_{k}+\Gamma_k \mathcal{D}_{2,k}\mathcal{L}_k))\tau_k-K_kv_k$ denotes the noise term and $o_{u,k}=-K_k\zeta_k-K_k\tilde{g}_{x,k}$ denotes uncertainties from the approximation error and modeling. 
	Consequently, the parameter error covariance matrix becomes 
	\begin{equation}
	\label{eq:P_K+1|K}
	\begin{split}
	&P_{k|k}^{\theta\theta}=\Pi_{k}P_{k-1|k-1}^{\theta\theta}\Pi_{k}^T+\mathbb{E}\{
	\Pi_{k}\epsilon_{k-1|k-1}o_{u,k}^T+o_{u,k}\epsilon_{k-1|k-1}^T\Pi_{k}^T\}
	+\mathbb{E}\{o_{n,k}o_{n,k}^T\}+\mathbb{E}\{o_{u,k}o_{u,k}^T\}+\Delta P_{k|k}
	\end{split}
	\end{equation}
	where $\Pi_{k}=(A_{k-1}+\Theta_{k-1}\mathcal{D}_{1,k-1}\mathcal{L}_{k-1})(I-K_k^\theta (B_k+\Gamma_k \mathcal{D}_{2,k}\mathcal{L}_k))$, which satisfy
	\begin{equation}
	\nonumber
	\begin{split}
	\|\Pi_{k}\|&\leq (a_{max}+\vartheta_{max}d_{1,max}l_{max})(1+\bar{k}_\theta(b_{max}+\gamma_{max}d_{2,max}l_{max})) \triangleq \bar{\pi}
	\end{split}
	\end{equation}
	where $\bar{k}_\theta$ denotes the upper bound of the gain with $\bar{k}_\theta \leq \lambda_{max}^{\theta}(b_{max}+\gamma_{max}d_{2,max}l_{max})/\delta_{v,max}$.
	
	By substituting $\epsilon_{k|k}$ into ${J}(\epsilon_{k|k})$, one obtains
	\begin{equation}
	\label{eq:constructor}
	\begin{aligned}
	{J}(\epsilon_{k|k})&=\epsilon_{k-1|k-1}^T\Pi_{k}^T(P_{k|k}^{\theta \theta})^{-1}\Pi_{k}\epsilon_{k-1|k-1}+\epsilon_{k-1|k-1}^T\Pi_{k}^T(P^{\theta \theta}_{k|k})^{-1}o_{u,k}+o_{u,k}^T(P^{\theta \theta}_{k|k})^{-1}\Pi_{k} \epsilon_{k-1|k-1}\\
	&+o_{u,k}^T(P^{\theta \theta}_{k|k})^{-1}o_{u,k}+o_{n,k}^T (P^{\theta \theta}_{k|k})^{-1}o_{n,k}\\
	\end{aligned}
	\end{equation}
	Each term in (\ref{eq:constructor}) can be shown to be bounded by utilizing certain conditions of the proposed assumptions. The simplified derivation process is as follows.
	
	Simplify (\ref{eq:P_K+1|K}) as $P_{k|k}^{\theta\theta}
	=\Pi_kP_{k-1|k-1}^{\theta\theta}\Pi_k^T+\Gamma_{k|k}^{*}$, then
	\begin{equation}
	\label{eq:Pk+1|k_simple_1}
	\begin{split}
	P_{k|k}^{\theta\theta}
	&=\Pi_k\left\{P_{k-1|k-1}^{\theta\theta}+\Pi_k^{-1}\Gamma_{k|k}^{*}\Pi_k^{-T}\right\}\Pi_k^{T}\\
	\end{split}
	\end{equation}
	where $\Pi_k^{-1}\Gamma_{k|k}^{*}\Pi_k^{-T}
	\ge {\delta_{\tau,max} }/{\bar{\pi}^2}$.
	
	It follows that the next inequality can be obtained:
	\begin{equation}
	\begin{split}
	&\epsilon_{k-1|k-1}^T\Pi_k^{T}(P_{k|k}^{\theta\theta})^{-1}\Pi_k\epsilon_{k-1|k-1}\le (1-\mu_p) \epsilon_{k-1|k-1}^T(P_{k-1|k-1}^{\theta\theta})^{-1}\epsilon_{k-1|k-1}
	\end{split}
	\end{equation}
	where $1-\mu_p=\left[1+{\delta_{\tau,max}}/{\bar{\pi}^2}\right]^{-1}$. 
	
	Given that
	\begin{equation}
	\label{eq:PE_fcn1}
	{J}(\epsilon_{k-1|k-1})=\epsilon_{k-1|k-1}^T (P_{k-1|k-1}^{\theta \theta})^{-1} \epsilon_{k-1|k-1}
	\end{equation}
	the first term in (\ref{eq:constructor}) can be shown to be bounded by $(1-\mu_p) {J}(\epsilon_{k-1|k-1})$.
	
	Regarding the uncertainty term $o_{u,k}$, we have
	\begin{equation}
	\label{eq:bound_ou}
	\|o_{u,k}\| \le \bar{k}_{\theta}(\bar{\zeta}+\bar{\tilde{g}}_x \varepsilon_{x}) \triangleq \bar{o}_u
	\end{equation}
	where the uncertainty term satisfies
	\begin{equation}
	\label{eq:uncertainty}
	o_{u,k}^T (P^{\theta \theta}_{k|k})^{-1}o_{u,k} \leq  \bar{o}_u^2/ {\lambda}_{max}^\theta 
	\end{equation}
	where $ {\lambda}_{max}^\theta  I$ denotes the upper bound for $P_{k|k}^{\theta \theta}$, which has been shown earlier. 
	
	Then $\epsilon_{k\!-\!1|k\!-\!1}^T\Pi_{k}^T(P^{\theta \theta}_{k|k})^{\!-\!1}o_{u,k}\!+\!o_{u,k}^T(P^{\theta \theta}_{k|k})^{\!-\!1}\Pi_{k} \epsilon_{k\!-\!1|k\!-\!1}$ is  consequently bounded by $2\bar{\pi}\bar{o}_u\|\epsilon_{k\!-\!1|k\!-\!1}\|/ {\lambda}_{max}^\theta  $. Considering
	$\|\epsilon_{k-1|k-1}\| \leq \epsilon_f$,
	it follows that the second and third terms of Eq. (\ref{eq:constructor}) can satisfy
	\begin{equation}
	\nonumber
	\begin{split}
	&\epsilon_{k-1|k-1}^T\Pi_{k}^T(P^{\theta \theta}_{k|k})^{-1}o_{u,k}+o_{u,k}^T(P^{\theta \theta}_{k|k})^{-1}\Pi_{k} \epsilon_{k-1|k-1}\leq 2\bar{\pi}\bar{o}_u\epsilon_f/ {\lambda}_{max}^\theta 
	\end{split}
	\end{equation}	
	Regarding the noise term $o_{n,k}$, we have
	\begin{equation}
	\label{eq:bound_on}
	\begin{split}
	&\|o_{n,k}o_{n,k}^T\| \le (1+\bar{k}_\theta(b_{max}+\gamma_{max}d_{max}l_{max}))^2\delta_{\tau,max}+\bar{k}_\theta^2\delta_{v,max}\triangleq \bar{o}_n
	\end{split}
	\end{equation}
	Therefore, the following inequality holds
	\begin{equation}
	\label{eq:noise}
	o_{n,k}^T (P^{\theta \theta}_{k|k})^{-1}o_{n,k} \leq  \bar{o}_n/ {\lambda}_{max}^\theta 
	\end{equation}
	Consequently, one can obtain
	\begin{equation}
	\label{eq:constructor_bound}
	\begin{aligned}
	&{J}(\epsilon_{k|k})\leq (1-\mu_p) {J}(\epsilon_{k-1|k-1})+\mu_q
	\end{aligned}
	\end{equation}
	where $\mu_q=2\bar{\pi}\bar{o}_u\epsilon_{f}/ {\lambda}_{max}^\theta  +\bar{o}_u^2/ {\lambda}_{max}^\theta +\bar{o}_n/ {\lambda}_{max}^\theta $.
	Therefore, the following inequality can be observed:
	\begin{equation}
	\label{eq:constructor_bounded}
	\begin{aligned}
	\mathbb{E}[{J}(\epsilon_{k|k})]-{J}(\epsilon_{k-1|k-1})&\leq-{\mu_p} {J}(\epsilon_{k-1|k-1})+\mu_q
	\end{aligned}
	\end{equation}
	where $\mu_q>0$ and $0 < \mu_p < 1$.
	
	The parameter estimation error $\epsilon_{k|k}$ in presence of bounded sensor modeling uncertainties by using the \textit{3rd-degree} CNF satisfies the root mean square boundedness, \textit{i.e.}, 
	\begin{equation}
	\begin{split}
	&E\{\|\epsilon_{k|k}\|^2\} \\
	&\le \frac{ {\lambda}_{max}^\theta }{ {\lambda}_{min}^\theta } E \{\|\epsilon_{0|0}\|^2\}(1-\mu_p)^{k}
	+\frac{\mu_q}{ {\lambda}_{min}^\theta }\sum_{i=1}^{k-1}(1-\mu_p)^i\\
	&\le \frac{ {\lambda}_{max}^\theta }{ {\lambda}_{min}^\theta } E \{\|\epsilon_{0|0}\|^2\}(1-\mu_p)^{k}
	+\frac{\mu_q}{ {\lambda}_{min}^\theta }\sum_{i=1}^{\infty}(1-\mu_p)^i\\
	&= \frac{ {\lambda}_{max}^\theta }{ {\lambda}_{min}^\theta } E \{\|\epsilon_{0|0}\|^2\}(1-\mu_p)^{k}
	+\frac{\mu_q}{ {\lambda}_{min}^\theta \mu_p}\\
	\end{split}
	\end{equation}
	when the initial error $\|\epsilon_{0|0}\|$ is bounded by  $\epsilon_f$. 
	
	Considering the Jensen's inequality, we have
	\begin{equation}
	\label{eq:jensen_par}
	\|\mathbb{E}\{(\epsilon_{k|k})^2\}\| \leq \mathbb{E}\{\|\epsilon_{k|k}\|^2\}
	\end{equation} 
	where the upper bound of the parameter estimation error can be given as
	\begin{equation}
	\label{eq:error_bound_par}
	\begin{split}
	&\|\mathbb{E}\{\epsilon_{k|k}\}\|\leq \sqrt{\mathbb{E}\{\|\epsilon_{k|k}\|^2\}}\leq \sqrt{\frac{ {\lambda}_{max}^\theta }{ {\lambda}_{min}^\theta } E \{\|\epsilon_{0|0}\|^2\}(1-\mu_p)^{k}
		+\frac{\mu_q}{ {\lambda}_{min}^{\theta} \mu_p}}
	\end{split}
	\end{equation}
	This completes the proof of the theorem.
	
\end{proof}
\textbf{Remark 7: } Compared to the existing estimation error boundedness analysis of CKF, our  analysis has the following unique features. First, distinct from \cite{zarei2015convergence}, which is based on a nonlinear system but with linear measurements, our boundedness analysis is conducted for nonlinear stochastic systems with nonlinear measurement equations. Analyzing the error boundedness  of the \textit{3rd-degree} CNF with nonlinear measurement expressions is more challenging than analyzing that of linear measurement equations due to the resulting higher approximation errors associated with the cubature. Second, distinct from the work presented in \cite{xu2016stochastic}, we further added the term ($\zeta(x_k, u_k)$) representing uncertainties into the boundedness analysis. Consequently, uncertainty from both the approximation error of the cubature rules as well as measurement uncertainties are taken into account and considered. Importantly, we have analyzed the interactive error effects $\tilde{g}_{x,k}$ from the state estimation, which has not been considered in the boundedness analysis of the relevant literature.  

Since the above boundedness analysis on parameter estimation error can be guaranteed with one important premise that the state estimation error and its covariance are bounded (i.e., Assumption 3), the following section will verify such rationality.

\subsection{Boundedness Analysis of State Estimation Error} 
\label{subsec:bound_state}
The goal of this section is to verify the boundedness of the state estimation error for achieving the boundedness of the estimated parameters. Considering the following reformulated nonlinear system for state estimation problem at time $k$ given the estimated $\hat{{\theta}}_{k-1}$
\begin{equation}
\label{eq:nonlinear_sys_state1}
{\Omega}_x:  
\left\{ 
\begin{array}{l}
x_k={f}(x_{k-1},\hat{\theta}_{k-1},u_{k-1})+\tilde{f}_{\theta,k-1}+w_k\\
y_k={g}(x_k,\hat{\theta}_{k-1},u_{k})+\tilde{g}_{\theta,k-1}+v_k\\
\end{array}
\right.
\end{equation}
where,
\begin{equation}
\nonumber
\begin{split}
\tilde{f}_{\theta,k-1}& \triangleq {f}(x_{k-1},{\theta}_{k-1},u_{k-1})-{f}(x_{k-1},\hat{\theta}_{k-1},u_{k-1})\\
\tilde{g}_{\theta,k-1}& \triangleq {g}(x_{k},{\theta}_{k-1},u_{k})-{g}(x_{k},\hat{\theta}_{k-1},u_{k})\\
\end{split}
\end{equation}
The error due to the estimate $\hat{\theta}_{k-1}$ is accounted for by introducing $\tilde{f}_{\theta,k-1}$ and $\tilde{g}_{\theta,k-1}$. 

Based on the Taylor series expansion of $\tilde{f}_{\theta,k-1}$ at $\hat{{\theta}}_{k-1|k-1}$, one obtains $\|\tilde{f}_{{\theta},{k-1}}\|=\|\alpha_{f_\theta}F_{\theta,k-1}\tilde{\theta}_{k-1|k-1}\|$,	where $\alpha_{f_\theta}=diag(\alpha_{f_\theta,1,k},\alpha_{f_\theta,2,k},\cdots,\alpha_{f_\theta,n_\theta,k})$ denotes an unknown instrumental diagonal matrix to compensate the high-order terms of expansion, which is assumed to satisfy $\|\alpha_{f_\theta}\| \leq \bar{\alpha}_{f_\theta}$. Given that the prior parameter error and its covariance matrix are assumed to be bounded at the time instant $k-1$, the inequality $\|\tilde{f}_{\theta,k-1}\|\leq \bar{\tilde{f}}_\theta \epsilon_{\theta}$ can be obtained, with $\bar{\tilde{f}}_\theta=\bar{\alpha}_{f_\theta}\bar{f}_{\theta}$, $\|F_{\theta,k-1}\| \leq \bar{f}_{\theta}$.

Similarly, $\|\tilde{g}_{\theta,k-1}\|\leq \bar{\tilde{g}}_\theta \epsilon_{\theta}$ can be also devised by using the Taylor series expansion of $\tilde{g}_{\theta,k-1}$ at $\hat{{\theta}}_{k-1|k-1}$ and each of the matrices can be bounded.

In this paper, we have developed 5th-degree cubature rule-based nonlinear filters for state estimation. Due to the fact that different 5th-degree cubature-based nonlinear filters can have different weights and cubature points, therefore, the boundedness  analysis differs from each other. Here the boundedness analysis will focus on the employed CNF-IV. 

Based on the Taylor series expansion of the nonlinear functions $f(\cdot)$ and $g(\cdot)$, we have
\begin{equation}
\label{eq:PE_k|k-1_x}
\epsilon_{k|k-1}^x=C_{k-1}\epsilon_{k-1|k-1}^x+\psi_x(\hat{x}_{k-1|k-1},x_{k-1})+\tilde{f}_{\theta,k-1}+w_{k}
\end{equation}
\begin{equation}
\label{eq:Delta_z_k|k-1_state}
\Delta z_{k|k-1}^x= D_k\epsilon_{k|k-1}^x+\zeta_k+\psi_z^{x}(\hat{x}_{k|k-1},x_k)+\tilde{g}_{\theta,k-1}+v_{k}
\end{equation}
\begin{equation}
\label{eq:PE_k|k_x}
\begin{split}
\epsilon_{k|k}^x&=(I-K_k^x D_k)\epsilon_{k|k-1}^x-K_k^x\psi_z^{x}(\hat{x}_{k|k-1},x_k)  \\
&-K_k^x \zeta_k-K_k^x\tilde{g}_{\theta,k-1}-K_{k}^x v_k\\
\end{split}
\end{equation}
where $C_{k-1}=\partial f(\cdot)/\partial x_{k-1}$ and $D_k=\partial g(\cdot)/\partial x_k$. $\psi_x(\hat{x}_{k-1|k-1},x_{k-1})$ and $\psi_z^{x}(\hat{x}_{k|k-1},x_k)$ represent the higher order terms which involve truncation errors associated with the approximation. This can be transformed into easy-to-handle formulations, i.e., $\psi_{x,k-1}=\mathcal{X}_{k-1}\mathcal{S}_{1,k-1}\mathcal{M}_{k-1}\epsilon_{k-1|k-1}^{x}$ and $\psi_{z,k}^{x}=\mathcal{Z}_k\mathcal{S}_{2,k}\mathcal{M}_k\epsilon_{k|k-1}^{x}$, where $\mathcal{X}_{k-1}$ and $\mathcal{Z}_k$ denote problem-dependent scaling matrices, $\mathcal{M}_{k}$ is introduced to provide an extra degree of freedom to tune the filter, $\mathcal{S}_{1,k-1}$ and $\mathcal{S}_{2,k}$ denote unknown time-varying matrices accounting for the linearization errors of the dynamical model which satisfy $\mathcal{S}_{1,k-1}\mathcal{S}_{1,k-1}^T\le I$ and $\mathcal{S}_{2,k}\mathcal{S}_{2,k}^T\le I$, respectively.
\vspace{1mm}\\
\textbf{\text{Assumption 5:}}
It is assumed that \textit{(a)} $C_kC_k^T \le c_{max}^2I$, \textit{(b)} the error covariance matrix satisfies $\sigma_{min} I \le P_{k-1|k-1}^{xx} \le \sigma_{max} I$, \textit{(c)} $\chi_{min}^2I \le \mathcal{X}_{k}\mathcal{X}_{k}^T \le \chi_{max}^2I$, \textit{(d)} the inequalities $ \Sigma_{w,k} \le \delta_{w,max}I$, where all the bounds are positive constants.

Let us define a performance index as ${J}(\epsilon_{k|k}^{x})=(\epsilon_{k|k}^x)^T (P_{k|k}^{xx})^{-1} \epsilon_{k|k}^x$, which can be expressed as
\begin{equation}
\label{eq:constructor_state}
\begin{aligned}
&{J}(\epsilon_{k|k}^{x})=(\epsilon_{k-1|k-1}^x)^T(\Pi_{1,k}\Pi_{2,k})^T(P_{k|k}^{xx})^{-1}\Pi_{1,k}\Pi_{2,k}\epsilon_{k-1|k-1}^x+(\epsilon_{k-1|k-1}^x)^T(\Pi_{1,k}\Pi_{2,k})^T(P^{xx}_{k|k})^{-1}o_{u,k}+o_{u,k}^T(P^{xx}_{k|k})^{-1} \\
&\times\Pi_{k} \epsilon_{k-1|k-1}^x+o_{u,k}^T(P^{xx}_{k|k})^{-1}o_{u,k}+o_{n,k}^T (P^{xx}_{k|k})^{-1}o_{n,k}\\
\end{aligned}
\end{equation}
where
\begin{equation}
\nonumber
\begin{split}
&\Pi_{1,k}^{x}=I-K_k^{x}(D_k+\mathcal{Z}_k\mathcal{S}_{2,k}\mathcal{M}_k)\\
&\Pi_{2,k}^{x}=C_{k-1}+\mathcal{X}_{k-1}\mathcal{S}_{k-1}\mathcal{M}_{k-1}\\
&o_{n,k}^x=\Pi_{1,k}^{x}w_k-K_k^xv_k\\
&o_{u,k}^x=\Pi_{1,k}^{x} \tilde{f}_{\theta,k-1}-K_k^x \tilde{g}_{\theta,k-1}-K_k^x  \zeta_k
\end{split}
\end{equation}
It is easy to observe the boundedness of the above variables, which are defined as $\|\Pi_{1,k}^{x}\| \leq \bar{\pi}_{1}$, $\|\Pi_{1,k}^{x}\|\leq \bar{\pi}_{2}$, $\|o_{n,k}^x\| \leq \bar{o}_{n}^x$ and $\|o_{u,k}^x\| \leq \bar{o}_{u}^x$.

Through tedious algebraic manipulations and assuming that $\|\epsilon_{k-1|k-1}^x\| \leq \epsilon_{x}$, each term in Eq. (\ref{eq:constructor_state}) can be shown to be bounded by utilizing certain conditions of the Assumption 5, which can be expressed as follows:
\begin{equation}
\label{eq:constructor_bound_state}
\begin{aligned}
&{J}(\epsilon_{k|k}^{x})\leq (1-\varrho_p) {J}(\epsilon_{k-1|k-1})+\varrho_q
\end{aligned}
\end{equation}
where, 
\begin{equation}
\nonumber
\begin{split}
&\varrho_p=1-\left[1+{\delta_{\tau,max}}/{\bar{\pi}_1^2\bar{\pi}_2^2}\right]^{-1}\\
&\varrho_q=2\bar{\pi}_{1}\bar{\pi}_{2}\bar{o}_u^x\epsilon_{x}/ \sigma_{max} +\{(\bar{o}_u^x)^2+\bar{o}_n^x\}/ \sigma_{max}
\end{split}
\end{equation}
with $\bar{o}_u^x=\bar{\pi}_{1}\bar{\tilde{f}}_{\theta}+\bar{k}_x\bar{\tilde{g}}_{\theta}+\bar{k}_x\bar{\zeta}$ and $\bar{o}_n^x=\bar{\pi}_{1}^2\delta_{w,max}+\bar{k}_x^2\delta_{v,max}$.
Therefore, the following inequality can be observed:
\begin{equation}
\label{eq:constructor_bounded_x}
\begin{aligned}
\mathbb{E}[{J}(\epsilon_{k|k}^x)]-{J}(\epsilon_{k-1|k-1}^x)&\leq-{\varrho_p} {J}(\epsilon_{k-1|k-1}^x)+\varrho_q
\end{aligned}
\end{equation}
where $\varrho_q>0$ and $0 < \varrho_p < 1$.

Based on the Assumption 5 and Eq. (\ref{eq:PE_fcn1}), one obtains
\begin{equation}
\label{eq:PE_fcn1_bound_state}
\frac{1}{\sigma_{max}}\|\epsilon_{k-1|k-1}^x\|^2 \le{J}(\epsilon_{k-1|k-1}^{x}) \le \frac{1}{\sigma_{min}}\|\epsilon_{k-1|k-1}^x\|^2
\end{equation}

The state estimation error $\epsilon_{k|k}^x$ in presence of error coupling effects from the parameter estimation satisfies the root mean square boundedness, \textit{i.e.}, 
\begin{equation}
\begin{split}
&E\{\|\epsilon_{k|k}^x\|^2\} 
\le \frac{ {\sigma}_{max} }{ {\sigma}_{min} } E \{\|\epsilon_{0|0}^x\|^2\}(1-\varrho_p)^{k}
+\frac{\varrho_q}{ {\sigma}_{min} \varrho_p}\\
\end{split}
\end{equation}
when the initial error $\|\epsilon_{0|0}^x\|$ is bounded by  $\varepsilon_{x_0}$. 

Considering the Jensen's inequality,  the upper bound of the state estimation bias can be given as
\begin{equation}
\label{eq:error_bound}
\begin{split}
\|\mathbb{E}\{(\epsilon_{k|k}^x)\}\|&\leq \sqrt{\mathbb{E}\{\|\epsilon_{k|k}^x\|^2\}}\leq \sqrt{\frac{ {\sigma}_{max} }{ {\sigma}_{min} } E \{\|\varepsilon_{x_0}\|^2\}(1-\varrho_p)^{k}
	+\frac{\varrho_q}{ {\sigma}_{min} \varrho_p}}
\end{split}
\end{equation}

Towards this end, our concerned parameter estimation error can be ultimately to be justified as bounded. Specifically, by assuming that both the initial errors and the initial error covariance matrices for states and parameters are bounded, the state estimation at next time instant can be bounded by the  condition in (\ref{eq:error_bound}). This condition will be subsequently utilized into the parameter estimation error boundedness analysis, leading to the bounded error that is expressed in (\ref{eq:error_bound_par}). Such analysis approach is motivated by the proposed dual estimation scheme in this paper. That is, the developed state filter and parameter filter are concurrently running, which indicates that one estimate is obtained and optimized at one time and then alternate to estimate the other. 
\\
\textbf{Remark 8:} It should be noted that the nonlinearities, faults and modeling uncertainties lead to the deviation of the possible equilibrium points. Therefore, we aim to consider the exponential boundedness in mean square (rather than the convergence) of the estimation error for both states and parameters. As shown in the proposed theorems, sufficient conditions under certain assumptions are given to achieve the desired performance requirements. Importantly, upper bounds on the estimation bias for the developed dual estimation scheme are provided, even taking into account the interactive error coupling effects between states and parameters. Further research directions include development of global convergence criterion for the joint state and fault estimation algorithm.



\section{Discussion and Conclusions} \label{sec7}
In this paper, a novel hybrid-degree dual estimation framework is proposed by using case-dependent cubature rules and our proposed cubature-based nonlinear filters for performing simultaneously state and parameter estimation objectives. The performance of our proposed hybrid-degree dual estimation strategy is demonstrated and evaluated by its application to a nonlinear gas turbine engine system for solving component fault diagnosis problem. From the perspective of dual estimation performance, our proposed hybrid-degree scheme with the \textit{5th-degree} for state estimation and the \textit{3rd-degree} for parameter estimation demonstrates its superiorities in terms of estimation accuracy and robustness to unmodelled dynamics and parametric uncertainties as compared to  cubature Kalman filters and unscented Kalman filters, and  computational efficiency as compared to the well-known particle filters.  The superiority, especially of the hybrid combination of Hybrid \{VI-I\}, is reflected  by the promptness in fault detection time, lower false alarm rates, reasonable fault identification accuracy, guarantee of computational efficiency and estimation error boundedness. By incorporating a modified cubature point propagation method into our proposed hybrid solution, the robustness capabilities against  modeling uncertainties can be improved in terms of lower false alarms.

The above characteristics justify and substantiate the observation that our proposed strategy is more suitable for the purpose of fault diagnosis of safety critical nonlinear systems that require lower fault detection times, lower false alarm rates, and accurate identification of the current health status. The limitations of using deterministic sampling and weighting in  cubature-based nonlinear filters suggest that considering more effective and adaptive tuning of free parameters may lead to a promising solution for improving the overall diagnostics and estimation performance, especially robustness with respect to  modeling uncertainties, model mismatches, and disturbances. In addition, another one of our future work will be concentrated on efficiently estimating the noise statistics to improve the adaptivity and robustness of the developed cubature-based nonlinear filters. This is motivated by the fact that it has been shown by some research work in the literature that incorporating the noise statistic estimator into the filtering process can actually in an adaptive manner adjust the noise tuning parameters. The verification and validation of our proposed results to a real gas turbine engine is another topic for our future research.

\bibliography{arXiv_syy}

\end{document}